%% file: ms.tex
\documentclass
[
a4paper,
11pt
]
{scrartcl}

\usepackage[utf8x]{inputenx}
\usepackage[T1]{fontenc}

\input{preamble.tex}

\begin{document}

\maketitle

\begin{abstract}
	\input{abstract.tex}
\end{abstract}	

\section{Introduction}
\label{sec:intro}
\input{introduction.tex}

\section{Preliminaries}
\label{sec:prelim}
\input{preliminaries.tex}

\section{Weighted bipartite independent set}
\label{sec:bipartite}
\input{bipartite.tex}

\section{Polynomial time tractable classes of graphs}
\label{sec:poly}
\input{polyt_computable_homs.tex}

\section{Homomorphisms of partially labelled graphs}
\label{sec:pinning}
\input{part_labelled_graphs_pinning.tex}

\section{Hardness for trees}
\label{sec:trees}
\input{treegadgets.tex}
\section{Dichotomy theorems}
\label{sec:main-thm}
\input{main_thm.tex}

\section{Composite Numbers}
\label{sec:composites}
\input{comp_num.tex}

\section{Acknowledgements}
The first author would like to thank Leslie Ann Goldberg and David Richerby for fruitful discussions during the early 
stages of this work.

\newpage
\bibliographystyle{plain}
\bibliography{bibliography}
	
\end{document}

%% file: preamble.tex
\title{Counting Homomorphisms to Trees \\ Modulo a Prime}
\author{Andreas G\"obel, J. A. Gregor Lagodzinski, Karen Seidel}
\date{\today}

\usepackage{lmodern}
\usepackage{datetime}
\usepackage{hyperref}
\usepackage{listings}
\usepackage[greek,english]{babel}

\usepackage{xparse} 

\usepackage{graphicx}
\usepackage{enumerate}
\usepackage{enumitem}
\usepackage{longtable}
\usepackage{amsfonts,amssymb,amstext,amsthm,amsmath}
\usepackage{mathabx}
\usepackage[mathscr]{euscript}
\usepackage{bbm}  
\usepackage{tikz}
\usetikzlibrary{decorations.pathmorphing}
\usetikzlibrary{decorations.pathreplacing} 
\usetikzlibrary{calc}
\usetikzlibrary{shapes.symbols}
\usetikzlibrary{arrows} 
\usetikzlibrary{decorations.markings}

\usepackage{color,soul}

\usepackage{calc}
\usepackage{filecontents}

\usepackage{a4wide} 

\usepackage{xcolor}

\definecolor{mycolor1}{RGB}{0,5,200}
\definecolor{mycolor2}{RGB}{0,150,5}

\newcommand{\agob}[1]{\marginpar{\tiny{agob: #1}}}

\DeclareMathOperator{\Aut}{Aut}
\DeclareMathOperator{\Orb}{Orb}

\DeclareMathOperator{\dom}{dom}

\DeclareMathOperator{\Bigop}{\bigoplus\nolimits^\mathnormal{p}}

\DeclareMathOperator*{\Crossp}{\sideset{}{^\mathnormal{p}}{\bigoplus}}

\newcommand{\classfont}{\mathsf}
\newcommand{\class}[1]{\mbox{{\(\classfont{\,#1}\)}}}
\newcommand{\pol}{\ensuremath{\class{P}}}
\newcommand{\parp}{\ensuremath{\#_2\class{P}}}

\newcommand{\np}{\ensuremath{\class{NP}}}
\newcommand{\fp}{\ensuremath{\class{FP}}}
\newcommand{\shp}{\ensuremath{\#\class{P}}}
\newcommand{\numpp}{\ensuremath{\#_p\class{P}}}
\newcommand{\numkp}{\ensuremath{\#_k\class{P}}}

\newcommand{\prb}[1]{\textsc{#1}}

\newcommand{\nkp}[1][k]{\ensuremath{\#_{#1}\class{P}}}
\newcommand{\nhcol}[1]{\ensuremath{\#\textsc{HomsTo}{#1}}}
\newcommand{\khcol}[1][H]{\ensuremath{\#_k\prb{HomsTo}{#1}}}
\newcommand{\csp}{\ensuremath{\textsc{CSP}}}
\newcommand{\kcsp}[1][k]{\ensuremath{\#_{#1}\csp}}

\newcommand{\numpsat}{\ensuremath{\#_p\prb{SAT}}}
\newcommand{\ksat}{\ensuremath{\#_k\prb{SAT}}}

\newcommand{\calG}{\ensuremath{\mathcal{G}}}
\newcommand{\calI}{\ensuremath{\mathcal{I}}}

\newcommand{\calF}{\ensuremath{\mathcal{F}}}

\newcommand{\kkhcol}[2]{\ensuremath{\#_{#1}\prb{HomsTo}{#2}}}
\newcommand{\kbis}[1]{\ensuremath{\#_{#1}\prb{BIS}}}
\newcommand{\conbis}[1]{\ensuremath{\#_{#1}\prb{ConBIS}}}

\DeclareMathOperator{\sat}{sat}

\newcommand{\eqclass}[1]{[\hspace{-0.2em}[{#1}]\hspace{-0.2em}]}

\newcommand{\Z}{\ensuremath{\mathbb{Z}}}	
\newcommand{\zp}{\ensuremath{\mathbb{Z}_p}}
\newcommand{\zk}{\ensuremath{\mathbb{Z}_k}}
\newcommand{\zsp}{\ensuremath{\mathbb{Z}^*_p}}

\newcommand{\ints}{\ensuremath{\mathbb{N}}}

\newcommand{\lweight}{\ensuremath{\lambda_\ell}}
\newcommand{\rweight}{\ensuremath{\lambda_r}}

\DeclareDocumentCommand \wISet {m o o}{
	\IfNoValueTF{#2}{	
		\IfNoValueTF{#3}{\ensuremath{Z_{\lweight, \rweight}(#1)}}{}
	}{
	\IfNoValueTF{#3}{}{\ensuremath{Z_{#2, #3}(#1)}}
	}
}

\renewcommand{\L}{\ensuremath{L}}
\newcommand{\R}{\ensuremath{R}}
\newcommand{\lpart}{\ensuremath{V_{\L}}}
\newcommand{\rpart}{\ensuremath{V_{\R}}}
\newcommand{\bipG}{\ensuremath{(\lpart, \rpart, E)}}

\DeclareDocumentCommand \Path {o o}{	
	\IfNoValueTF{#1}{	
		\IfNoValueTF{#2}{\ensuremath{Q}}{\ensuremath{Q_{(#1)}}}	
	}{
	\IfNoValueTF{#2}{\ensuremath{Q_#1}}{\ensuremath{Q_{#1}}}}	
}

\newcommand{\lNeigh}{\ensuremath{W_{\L}}} 
\newcommand{\rNeigh}{\ensuremath{W_{\R}}} 

\newcommand{\IO}[1]{\mathfrak{#1}}	
\newcommand{\rel}[2]{#1 \sim_{\IO{I}} #2}	
\newcommand{\IOclass}[1]{\ensuremath{\eqclass{#1}}_\IO{I} } 

\newcommand{\parhcol}[1][H]{\ensuremath{\#_2\textsc{HomsTo}{#1}}}
\newcommand{\phcol}[1][H]{\ensuremath{\#_p\prb{HomsTo}{#1}}}
\newcommand{\hcol}[1][H]{\ensuremath{\prb{HomsTo}{#1}}}
\newcommand{\partlabphcol}{\ensuremath{\#_p \prb{PartLabHomsTo}H}}
\newcommand{\partlabparhcol}{\ensuremath{\#_2 \prb{PartLabHomsTo}H}}
\newcommand{\pbislr}[2]{\ensuremath{\#_p\prb{BIS}_{#1,#2}}}  

\newcommand{\pcsp}{\ensuremath{\#_p\csp}}

\newcommand{\Homs}[2]{\ensuremath{\mathrm{Hom}\left({#1}\to {#2}\right)}}
\newcommand{\isoto}{\cong}

\newcommand{\Lovasz}{Lov\'asz}
\newcommand{\Nesetril}{Ne\v{s}et\v{r}il}

\newcommand\probpar[4]{
 \begin{problem}
  \begin{description}
    \item \emph{Name.} #1
    \item \emph{Parameter.} #2
    \item \emph{Input.} #3
    \item \emph{Output.} #4

  \end{description}
  \end{problem}
}

\newcommand{\paris}{\ensuremath{\#_2\prb{IS}}}
\newcommand{\Homst}[2]{\mathrm{Hom}({#1}\to {#2})}

\newcommand{\dist}{\ensuremath{d}}
\renewcommand{\Aut}{\mathrm{Aut}}
\renewcommand{\Orb}{\mathrm{Orb}}

\newcommand{\vecu}{\mathbf{u}}
\newcommand{\vecv}{\mathbf{v}}
\newcommand{\vecw}{\mathbf{w}}
\newcommand{\ubar}{\bar{u}}
\newcommand{\vbar}{\bar{v}}
\newcommand{\wbar}{\bar{w}}
\newcommand{\xbar}{\bar{x}}

\newcommand{\InjHoms}[2]{\mathrm{InjHom}({#1}\to {#2})}

\newcommand{\kis}{\ensuremath{\#_k\prb{IS}}}

\newcommand{\seqclass}[1]{[\hspace{-0.27em}[{#1}]\hspace{-0.28em}]}

\theoremstyle{plain}
\newtheorem{theorem}{Theorem}[section]
\newtheorem{corollary}[theorem]{Corollary}
\newtheorem{lemma}[theorem]{Lemma}
\newtheorem{proposition}[theorem]{Proposition}
\theoremstyle{definition}
\newtheorem{definition}[theorem]{Definition}

\newtheorem{conjecture}[theorem]{Conjecture}

\newtheorem{problem}[theorem]{Problem}

%% file: abstract.tex
Many important graph theoretic notions can be encoded as counting graph homomorphism problems, such as partition 
functions in statistical physics, in particular independent sets and colourings.
In this article we study the complexity of~$\#_p\textsc{HomsTo}H$, the problem of counting graph homomorphisms from an input graph to a graph $H$ modulo a prime number~$p$.
Dyer and Greenhill proved a dichotomy stating that the tractability of non-modular counting graph homomorphisms depends on the structure of the target graph.
Many intractable cases in non-modular counting become tractable in modular counting due to the common phenomenon of cancellation.
In subsequent studies on counting modulo~$2$, however, the influence of the structure of~$H$ on the tractability was shown to persist, which yields similar dichotomies.

Our main result states that for every tree~$H$ and every prime~$p$ the problem $\#_p\textsc{HomsTo}H$ is either polynomial time computable or $\#_p\mathsf{P}$-complete.
This relates to the conjecture of Faben and Jerrum stating that this dichotomy holds for every graph $H$ when counting modulo~2.
In contrast to previous results on modular counting, the tractable cases of $\#_p\textsc{HomsTo}H$ are essentially the same for all values of the modulo when $H$ is a tree.
To prove this result, we study the structural properties of a homomorphism.
As an important interim result, our study yields a dichotomy for the problem of counting weighted independent sets in a bipartite graph modulo some prime~$p$.
These results are the first suggesting that such dichotomies hold not only for the one-bit functions of the modulo~2 case but also for the modular counting functions of all primes~$p$.

%% file: introduction.tex
Graph homomorphisms generate a powerful language expressing important notions; examples include constraint satisfaction 
problems and partition functions in statistical physics. As such, the computational complexity of graph homomorphism 
problems has been studied extensively from a wide range of views.
Early results include that of Hell and \Nesetril{}~\cite{HN90}, who study the complexity of~\hcol{}, the problem 
of deciding if there exists a homomorphism from an input graph $G$ to a fixed graph $H$. They show the following 
dichotomy: if $H$ is bipartite or has a loop, the problem is in $\pol$ and in every other case \hcol{} is 
$\np$-complete. In particular, this is of interest since a result of Ladner~\cite{Lad75} shows that if $\pol\neq\np$, then there 
exist problems that are neither in $\pol$ nor $\np$-hard.

Dyer and Greenhill~\cite{DG00} show a dichotomy for the problem \nhcol{H}, the problem of counting the 
homomorphisms from an input graph~$G$ to~$H$. Their theorem states that \nhcol{H} is tractable if $H$ is a complete 
bipartite graph 
or a complete graph with loops on all vertices; otherwise \nhcol{H} is $\shp$-complete. This dichotomy was 
progressively extended 
to weighted sums of homomorphisms with integer weights, by Bulatov and Gohe~\cite{BG05}; with real weights, by Goldberg 
et al.~\cite{GGJT}; finally, with complex weights, by Cai, Chen and Lu~\cite{CCL13}. 

We study the complexity of counting homomorphisms modulo a prime~$p$. The set of homomorphisms from 
the input graph~$G$ to the target graph~$H$ is denoted by $\Homs{G}{H}$. For each pair of fixed parameters $p$ and $H$,
 we study the computational problem \phcol{}, that is the problem 
of computing $|\Homs{G}{H}|$ modulo~$p$. The value of~$p$ and the structure of the target 
graph~$H$ influence the complexity of \phcol{}. Consider the graph $H$ in Figure~\ref{fig:intro_tree}. Our results show 
that \phcol{} is computable in polynomial time when $p=2,3$ while it is hard for any other prime~$p$.

\begin{figure}[t]
	\centering
	\includegraphics[width=0.3\textwidth]{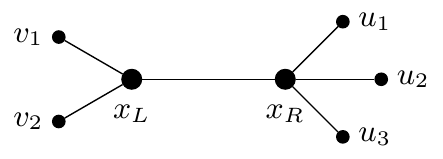}
	\caption{The graph $H$ will be our recurring example and the labelling of the vertices is justified later 
in the introduction.}
	\label{fig:intro_tree}
\end{figure}

Our main goal is to fully characterise the complexity of \phcol{} in a dichotomy theorem. In this manner we aim to 
determine for which pair of parameters~$(H,p)$ the problem is tractable and show that for every other pair of 
parameters the problem is hard. As the theorem of Ladner~\cite{Lad75} extends to the modular counting problems, it is 
not obvious that there are no instances of \phcol{} with an intermediate complexity.

The first study of graph homomorphisms under the setting of modular counting has been conducted by Faben and 
Jerrum~\cite{FJ13}. Their work is briefly described in the following and we assume the reader to be familiar 
with the notion of an automorphism and its order. We provide the formal introduction in Section~\ref{sec:prelim}.
Given a graph $H$ and an automorphism~$\varrho$ of $H$, $H^\varrho$ denotes the subgraph of $H$ induced by the 
fixpoints of $\varrho$.
We write $H \Rightarrow_k H'$ if there is an automorphism~$\varrho$ of order~$k$ of~$H$ such that $H^\varrho=H'$ and we 
write $H \Rightarrow_k^* H'$ if either $H$ is isomorphic to~$H'$ (written
$H\isoto H'$) or, for some positive integer~$t$, there are graphs $H_1, \dots, H_t$ such that 
$H \isoto H_1$, $H_1 \Rightarrow_k \cdots \Rightarrow_k H_t$, and $H_t  \isoto H'$.

Faben and Jerrum showed \cite[Lemma~3.3]{FJ13} that if the order of $\varrho$ is a prime~$p$, then
$|\Homs{G}{H}|$ is equivalent to $|\Homs{G}{H^\varrho}|$ modulo~$p$. 
Furthermore, they showed~\cite[Theorem~3.7]{FJ13} that there is (up to isomorphism) 
exactly one graph~$H^{*p}$ without automorphisms of order~$p$, such that  $H\Rightarrow_p^* H^{*p}\!$. This 
graph~$H^{*p}$ is called the \emph{order~$p$ reduced form} of~$H$. If $H^{*p}$ falls
into the polynomial computable cases of the theorem of Dyer and Greenhill, then 
\phcol{} is computable in polynomial time as well. For $p=2$, Faben and Jerrum conjectured that these are 
the only instances computable in polynomial time.

\pagebreak

\begin{conjecture}[Faben and Jerrum~\cite{FJ13}] \label{conj:FJ}
  Let $H$ be a  graph.  If its order~2 reduced form~$H^{*2}$ has at most one vertex, then \parhcol{} is in
    $\fp$; otherwise, \parhcol{} is $\parp$-complete.
\end{conjecture}

Faben and Jerrum~\cite[Theorem~3.8]{FJ13} underlined their conjecture by proving it for the case in which~$H$ is a 
tree. In subsequent works this proof was extended to cactus graphs in~\cite{GGR14} and to square-free 
graphs in~\cite{GGR15}, by G\"obel, Goldberg and Richerby. 

The present work follows a direction orthogonal to the aforementioned. Instead 
of proving the conjecture for richer classes of graphs, we show a dichotomy for all primes, starting again by restricting the target graph $H$ to be a tree.

\begin{theorem}\label{thm:modp-trees}
 Let $p$ be a prime and let $H$ be a graph, such that its order~$p$ reduced form~$H^{*p}$ is a tree. If $H^{*p}$ is a 
star, then \phcol{} is computable in polynomial time; otherwise, \phcol{} is $\numpp$-complete.
\end{theorem}

ur results are the first to suggest that the conjecture of Faben and Jerrum might apply to counting graph 
homomorphisms modulo every prime~$p$ instead of counting modulo~2. This suggestion, however, remains hypothetical. 
Borrowing the words of Dyer, Frieze and Jerrum~\cite{DFJ02}: ``One might even rashly conjecture'' it ``(though we shall 
not do so)''.

To justify our title we give the following corollary, stating a dichotomy for all trees~$H$.

\begin{corollary} \label{cor:modp-trees}
 Let $p$ be a prime and let $H$ be a tree. If the order~$p$ reduced form~$H^{*p}$ of $H$ is a star, then \phcol{} is 
computable in polynomial time; otherwise, \phcol{} is $\numpp$-complete.
\end{corollary}

We illustrate Theorem~\ref{thm:modp-trees} using the following discussion on Figure~\ref{fig:intro_tree}. The order~2 
and the order~3 reduced form of $H$ both are the graph with one vertex, whereas for any other prime the graph stays as 
such.

The polynomial computable cases follow directly from the results of Faben and Jerrum. Thus, to prove Theorem~\ref{thm:modp-trees} it suffices to show that \phcol{} is $\numpp$-complete for every tree~$H$ that is not a star and has no automorphism of order~$p$. The reductions in~\cite{FJ13,GGR14,GGR15} show hard instances of \parhcol{} by starting from~\paris, the problem of 
computing~$|\calI(G)|\pmod 2$, where $\calI(G)$ is the set of independent sets of $G$. \paris{} was shown to be $\parp$ 
complete by Valiant~\cite{Val06}. Later on, Faben~\cite{Fab08} extended this result by proving~\kis{} to be 
$\numkp$-complete for all integers~$k$. For reasons to be explained in Section~\ref{sec:intro-our-technique}
we do not use this problem as a starting point for our reductions. 

We turn our attention to $\kbis{p}$, the problem of counting the independent sets of a bipartite graph 
modulo~$p$. In the same work Faben~\cite{Fab08} includes a construction to show hardness for $\kbis{p}$. We employ the 
weighted version~\pbislr{\lweight}{\rweight} as a starting point for our reduction extending the 
research on~\kbis{}.

\probpar{\pbislr{\lweight}{\rweight}\label{prob:bislr}.}
	{$p$ prime and $\lweight, \rweight \in\Z_p$.}
	{Bipartite graph $G=\bipG$.}
	{$\wISet{G}=\sum_{I\in\calI(G)}\lweight^{|\lpart \cap I|}\rweight^{|\rpart \cap I|} \pmod p$.}

\smallskip
In fact, we obtain the following dichotomy.
	
\begin{theorem}\label{thm:wbis-hardness}
Let $p$ be a prime and let $\lweight$, $\rweight \in\zp$.
If $\lweight \equiv 0 \pmod p$ or $\rweight \equiv 0 \pmod p$, then \pbislr{\lweight}{\rweight} is computable in polynomial time.
Otherwise, \pbislr{\lweight}{\rweight} is \numpp{}-complete.
\end{theorem}

In order to prove hardness for \phcol{} we employ a reduction in three phases:
(i) we reduce the ``canonical'' $\numpp$-complete problem~\numpsat{} to \pbislr{\lweight}{\rweight}; (ii) we reduce \pbislr{\lweight}{\rweight} to \partlabphcol{}, a restricted version of \phcol{} which we define in Section~\ref{sec:intro-our-technique}; (iii) we reduce \partlabphcol{} to \phcol{}.

Section~\ref{sec:intro-mod-count} provides background knowledge on modular counting. In 
Section~\ref{sec:intro-related-work} we will discuss some related work. A high level proof of our three way reduction is provided in Section~\ref{sec:intro-our-technique}. There we also explain the technical obstacles arising from values of the modulo~$p>2$ and how we overcome them by generalising the techniques used for the case $p=2$. First, we explain step~(i), the reduction from \numpsat{} to \pbislr{\lweight}{\rweight}. Afterwards, we describe step~(iii), the reduction from \partlabphcol{} to \phcol{} establishing the required notation for the subsequent illustration of step~(ii), the reduction from \pbislr{\lweight}{\rweight} to \partlabphcol{}. 
In Section~\ref{sec:intro-composites} we discuss the limits of our techniques, which do not yield a dichotomy modulo any integer~$k$.

\subsection{Modular counting}
\label{sec:intro-mod-count}
Modular counting was originally studied from the decision problem's point of view. Here, the objective is to determine 
if the number of solutions is non-zero modulo~$k$. The complexity class $\oplus\pol$ was first studied by 
Papadimitriou and Zachos~\cite{PZ82} and by Goldschlager and Parberry~\cite{GP86:Parallel}. $\oplus\pol$ consists of all
problems of the form ``is $f(x)$ odd or even?'', where $f(x)$ is a function in $\shp$. A result of 
Toda~\cite{Tod91} states that every problem in the polynomial time hierarchy reduces in polynomial time to some problem 
in~$\oplus\pol$. This result suggests that $\oplus\pol$-completeness represents strong evidence for intractability.

For an integer $k$ the complexity class~$\nkp$ consists of all problems of the form 
``compute $f(x)$ modulo~$k$'', where $f(x)$ is a function in~$\shp$. In the special case of $k=2$, $\parp=\oplus\pol$, 
as the instances of $\parp$ require a one bit answer. Throughout this paper though, instead of the more traditional 
notation $\oplus\pol$, we will use $\parp$ to emphasise our interest in computing functions.

If a counting problem can be solved in polynomial time, the corresponding decision and modular counting problems can 
also be solved in polynomial time. The converse, though, does not necessarily hold. The reason is that efficient counting 
algorithms rely usually on an exponential number of cancellations that occur in the problem, e.g. compute the 
determinant of a non-negative matrix. The modulo operator introduces a natural setting for such cancellations to 
occur.

For instance, consider the $\shp$-complete problem of counting proper $3$-colourings of a graph~$G$ in the modulo~$3$ 
(or even modulo~$6$) setting. $3$-colourings of a graph assigning all three colours can be grouped in sets of 
size~$6$, since there are $3!=6$ permutations of the colours. Thus, the answer to these instances is always a multiple 
of $6$, and therefore ``cancels out''. It remains to compute the number of $3$-colourings assigning less 
than $3$ colours. For the case of using exactly $2$ colours we distinguish the following two cases: $G$ is not 
bipartite and there are no such colourings; $G$ is bipartite and the number of $3$-colourings of $G$ that use 
exactly $2$ colours is $3(2^c)$, where $c$ is the number of components of $G$.
Finally, computing the number of proper $3$-colourings of 
$G$ that use exactly one colour is an easy task. Either $G$ has an edge and there are no such colourings, or
$G$ has no edges and for every vertex there are $3$ colours to choose from.

Valiant~\cite{Val06} observed a surprising phenomenon in the tractability of modular counting problems. He showed that 
for a restricted version of $3$-SAT computing the number of solutions modulo~$7$ is in $\fp$, but computing 
this number modulo~$2$ is $\parp$-complete. This mysterious number~$7$ was later explained by Cai and 
Lu~\cite{CL11:Holographic}, who showed that the $k$-SAT version of Valiant's problem is tractable modulo any prime 
factor of $2^k-1$.

\subsection{Related work}
\label{sec:intro-related-work}
We have already mentioned earlier work on counting graph homomorphisms. In 
this section we highlight the work of Faben~\cite{Fab08} and the work of Guo et al.~\cite{GHLX11} on the complexity of the modular counting variant 
of the constraint satisfaction problem. 

\begin{problem}\label{prb:pcsp}
\begin{description}
\item \emph{Name.} $\kcsp[k](\calF)$.
\item \emph{Parameter.} $k\in\Z_{>0}$ and a set of functions $\calF=\{f_1,\dots, f_m\}$, where for each $j\in[m]$, $f_j:\{0,1\}^{r_j}\rightarrow \zp$ and $r_j\in\Z_{>0}$.
\item \emph{Input.} Finite set of constraints over Boolean variables $x_1,\dots, x_n$ of the form \\
$f_{j_l}(x_{i_{l,1}}, x_{i_{l,2}},\dots,x_{i_{l,r_{j_l}}})$.
\item \emph{Output.}
$\sum_{x_1,\dots,x_n\in\{0,1\}}\prod_l f_{j_l}(x_{i_{l,1}}, x_{i_{l,2}},\dots,x_{i_{l,r_{j_l}}})\pmod k$.
\end{description}
\end{problem}
Faben showed a dichotomy theorem~\cite[Theorem~4.11]{Fab08} when the functions in~$\calF$ have Boolean domain and Boolean range, i.e. $f:\{0,1\}\rightarrow\{0,1\}$. Guo et al. extended this dichotomy~\cite[Theorem~4.1]{GHLX11} to \kcsp{}, when the functions in~$\calF$ have Boolean domain $\{0,1\}$ but range in $\zk$.

Constraint satisfaction problems generalise graph homomorphism problems, when the domain of the constraint functions is arbitrarily large. In order to illustrate that \kcsp{} is a generalisation of \khcol{}, let $G$ be an input for 
\khcol{}, for which we describe an equivalent \kcsp{} instance. The domain of the constraint satisfaction problem 
is $D=V(H)$ and $\calF$ contains a single binary relation $R_H$, with $R_H(u,v)=1$ whenever $(u,v)\in E(H)$ and 
$R_H(u,v)=0$ otherwise. Thus, \khcol{} is an instance of $\kcsp(\{R_H\})$. The input of $\kcsp(\{R_H\})$  
contains a variable~$x_v$ for every vertex $v\in V(G)$ and a constraint~$R_H(x_u,x_v)$ for every edge $(u,v)\in E(G)$. 
As can be observed from the construction, every valid homomorphism $\sigma:V(G)\rightarrow V(H)$ corresponds to an 
assignment of the variables $\{x_v\}_{v\in V(G)}$ satisfying every constraint in the \csp{}.

The results of Faben and Guo et al. are incomparable to ours. We consider prime values of the modulo and a single 
binary relation, however the domain of our relations is arbitrarily large. Furthermore, the results of Faben~\cite[Theorem~4.11]{Fab08} show that the constraint language~$\calF$ for which \kcsp[2] is tractable is richer than the constraint language for which \kcsp[k] is tractable, where $k>2$. In contrast, our results show that the dichotomy criterion of \phcol{} remains the same for all primes~$p$, when $H$ is a tree.

\subsection{Beyond one-bit functions}
\label{sec:intro-our-technique}

\paragraph{Weighted bipartite independent sets}
To explain how we prove Theorem~\ref{thm:wbis-hardness}, consider a bipartite graph $G=(\lpart,\rpart,E)$ and let $\lweight=0$ (the case $\rweight=0$ is symmetric). We observe 
that every independent set~$I$ which contributes a non-zero summand to $\wISet{G}$ can only contain vertices 
in~$\rpart$ ($\wISet{G}$ is defined in Problem~\ref{prob:bislr}). This yields the closed form 
$\wISet{G}=(\rweight+1)^{|\rpart|}$, which is computable in polynomial time. Regarding the case 
$\lweight,\rweight\not\equiv 0 \pmod p$, we employ a generalisation of a 
reduction used by Faben. In~\cite[Theorem~3.7]{Fab08} Faben reduces \numpsat{} to~\pbislr{1}{1}, the problem of counting independent sets of a bipartite graph. 

We have to generalise this reduction for the weighted setting, in particular allowing different vertex weights for the
vertices of each partition. Furthermore, during the construction we have to keep track of the assignment of vertices to 
their corresponding part, $\lpart$ or $\rpart$. For this purpose we need to show the existence of bipartite graphs~$B$, 
where $\wISet{B}$ takes specific 
values. These graphs are then used as gadgets in our reduction. In the unweighted setting~\pbislr{1}{1} the 
graphs~$B$ are complete bipartite graphs. However, in the weighted setting~\pbislr{\lweight}{\rweight} complete 
bipartite graphs are not sufficient. Therefore, we prove the existence of the necessary bipartite gadgets~$B$ 
constructively. The technical proofs appear in Section~$3$.

\paragraph{Pinning}
Similar to the existing hardness proofs on modular counting graph homomorphisms we deploy a ``pinning'' technique.  A 
partial function from a set~$X$ to a set~$Y$ is a function $f:X'\rightarrow Y$ for some $X'\subseteq X$. For any 
graph~$H$, a \emph{partially $H$-labelled graph} $J=(G,\tau)$ consists of an \emph{underlying graph}~$G$ and a 
\emph{pinning function}~$\tau$, which is a partial function from $V(G)$ to~$V(H)$. A homomorphism from a partially 
labelled graph~$J=(G,\tau)$ to~$H$ is a homomorphism $\sigma\colon G\to H$ such that, for all vertices $v\in 
\dom(\tau)$, $\sigma(v) = \tau(v)$. The resulting problem is denoted by~\partlabphcol{}, that is, given a prime~$p$ and 
graph~$H$, compute $|\Homs{J}{H}| \pmod p$. In Section~\ref{sec:pinning}, we show that \partlabphcol{} reduces to 
\phcol{}. This allows us to establish hardness for \phcol{}, by proving hardness for \partlabphcol{}. The reduction 
generalises the pinning reduction of G\"obel, Goldberg and Richerby~\cite{GGR15} from \partlabparhcol{} to \parhcol{}. 

We explain how to prove pinning when we restrict the value of the modulo to~2 and the pinning function $\tau(J)=\{u\mapsto v\}$ to ``pin'' a single 
vertex. Given two graphs with distinguished vertices $(G,u)$ and~$(H,v)$, let $\Homs{(G,u)}{(H,v)}$ be the set of homomorphisms from $G$ to $H$ mapping $u$ to $v$.
Given a graph with a distinguished vertex $(G,u)$ and a graph $H$, we define~$\vecw_H(G)$ to be the $\{0,1\}$-vector 
containing the entries $|\Homs{(G,u)}{(H,v)}|\pmod 2$ for each vertex $v\in V(H)$. Observe that for two vertices 
$v_1,v_2\in V(H)$, such that $(H,v_1)\isoto(H,v_2)$, and any graph $G$ the relevant entries in $\vecw_G(H)$ will 
always be equal. Therefore, we can contract all such entries to obtain the \emph{orbit vectors} $\vecv_H(G)$. 
Suppose that there exists a graph with a distinguished vertex $(\Theta,u_\Theta)$, such that 
$\vecv_H(\Theta)=0\dots010\dots0$, where the 1-entry corresponds to the vertex $v$ of $H$. Given our input $J$ for 
\partlabparhcol{} we can now define an input $G$ for \parhcol{}, such that $|\Homs{J}{H}|\equiv|\Homs{(G(J),u)}{(H,v)}| \equiv|\Homs{G}{H}|\pmod2$. 
$G$ contains a disjoint copy of $G(J)$ and $\Theta$, where the vertices $u$ and $u_\Theta$ are identified (recall that 
$u$ is the vertex of $J$ mapped by $\tau(J)$). Due to the value of $\vecv_H(\Theta)$ and the structure of $G$ there is 
an even number of homomorphisms mapping $u$ to any vertex $v'\neq v$, which establishes the claim.

Such a graph $\Theta$, however, is not guaranteed to exist. Instead, we can define a set of operations on the vectors 
$\vecv_H$ corresponding to graph operations and show that for any vector in $\{0,1\}^{|V(H)|}$ there exist a 
sequence of graphs with distinguished vertices $(\Theta_1, u_1), \dots , (\Theta_t, u_t)$ that ``generate'' this vector. 
Thus, there exists a set of graphs that ``generate'' $\vecv=0\dots010\dots0$, which yields the desired reduction. This
technique of~\cite{GGR15} exploits the value of the modulo to be~2. Applying this technique to counting modulo any 
prime~$p$ directly, one can establish pinning for asymmetric graphs, that is graphs whose automorphism group contains only the identity. A dichotomy for 
\phcol{}, when $H$ is an asymmetric tree appears in the first author's doctoral thesis~\cite{GobThesis}.

In order to go beyond asymmetric graphs, one has to observe that information redundant only in the modulo~2 case is 
lost from the contraction of the vectors~$\vecw_H$ to the vectors~$\vecv_H$. This works on asymmetric graphs, since then
these two vectors are identical. For general graphs we are able to restore pinning for counting homomorphisms modulo any prime~$p$ by utilising the non-contracted vectors~$\vecw_H$.

\begin{theorem}\label{thm:partlabcol}
Let $p$ be a prime and let $H$ be a graph.
Then \partlabphcol{} reduces to \phcol{} via polynomial time Turing reduction.
\end{theorem}

To obtain hardness for \phcol{} we only need to pin two vertices when~$H$ is a tree, i.e. the domain of the pinning function $\tau$ has size two. For a study of a more general class of target graphs~$H$ (see~\cite{GGR15}), the size of the domain has to be larger. As our pinning theorem applies to all primes~$p$, all graphs~$H$ and pinning functions of arbitrary domain size, it can potentially be used to show hardness for \phcol{} for all primes and any class of target graphs~$H$. The formal proofs appear in Section~$5$.

\paragraph{Gadgets}
Gadgets are structures appearing in the target graph $H$ that allow to reduce \paris{} to \partlabparhcol{} (the 
hardness of \parhcol{} is then immediate from Theorem~\ref{thm:partlabcol}). 
For illustrative purposes we simplify the definitions appearing in~\cite{GGR15}. 
\parhcol{}--gadgets consist of two partially labelled graphs with distinguished vertices $(J_1,y)$, $(J_2,y,z)$ 
along with two ``special'' vertices $i,o\in V(H)$. Given the input~$G$ for \paris{}, we construct an input~$G'$ for
\partlabparhcol{} as follows. We attach a copy of $J_1$ to every vertex~$u$ of $G$ (identifying $u$ with $y$) and 
replace every edge~$(u,v)$ of $G$ with a copy of $J_2$ (identifying $u$ with $y$ and $v$ with $z$). 
The properties of $J_1$ ensure that there is an odd number of homomorphisms from $G'$ to $H$ where the original 
vertices of $G$ are mapped to $i$ or $o$, while the number of the remaining homomorphisms cancels out. The properties 
of $J_2$ ensure that there is an even number of homomorphisms from $G'$ to $H$ when two adjacent vertices of $G$ are 
both mapped to $i$, and an odd number of homomorphisms in every other case. We can now observe that $|\calI(G)|\equiv 
|\Homs{G'}{H}|\pmod2$, as the set of homomorphisms that does not cancel out must map every vertex of $G$ to $i$ or $o$ 
and no pair of adjacent vertices both to $i$. Every vertex of $G$ that is in an independent set must be mapped to~$i$, 
and every vertex that is out of the independent set must be mapped to $o$.

Generalising the described approach to modulo any prime $p>2$ one would end up reducing from a restricted \pcsp{} instance, containing a binary relation and a unary weight 
that must be applied to every variable of the instance (this is known as \emph{external field} in statistical 
physics). Similar to the modulo~2 case the edge interaction is captured by the binary relation and size of the set of ``special'' vertices by the unary weights. Since for primes $p>2$ there are more non-zero values than 1 (odd) a study of the external field is no longer trivial in this case.
Instead we choose a different approach and reduce from \pbislr{\lweight}{\rweight}. This seems to capture the 
structure that produces hardness in \phcol{} in a more natural way.

We formally present our reduction in Section~\ref{sec:trees}. In the following we sketch our proof method and focus our 
attention on the example graph~$H$ in Figure~\ref{fig:intro_tree}. Let $G=(\lpart,\rpart,E)$ be a bipartite graph. 
Homomorphisms from $G$ to $H$ must respect the partition of $G$, i.e. the vertices in $\lpart$ can only be mapped to 
the vertices in $\{x_L,u_1,u_2,u_3\}$ and the vertices in $\rpart$ can only be mapped to the vertices in $\{x_R, 
v_1,v_2\}$, or vice versa. Any homomorphism~$\sigma$ from~$G$ to~$H$, which maps the vertex $w\in V(G)$ to any vertex in 
$\{u_1,u_2,u_3\}$, must map every neighbour of~$w$ to~$x_R$. Similarly, any homomorphism~$\sigma$ from~$G$ to~$H$, which 
maps the vertex $w\in V(G)$ to any vertex in $\{v_1,v_2\}$, must map every neighbour of~$w$ to~$x_L$. Thus, 
homomorphisms from~$G$ to~$H$ express independent sets of $G$: $\{u_1,u_2,u_3\}$ represent the vertices of $\lpart$ in 
the independent set and $\{v_1,v_2\}$ represent the vertices of $\rpart$ in the independent set, or vice versa. We 
construct a partially labelled graph $J$ from $G$ to fix the choice of $\lpart$ and $\rpart$ in the set of 
homomorphisms from $G$ to $H$. $G(J)$ contains a copy of $G$, where every vertex in $\lpart$ is attached to the new 
vertex $\hat{u}$ and every vertex in $\rpart$ is attached to the new vertex $\hat{v}$. In addition, 
$\tau(J)=\{\hat{u}\mapsto 
x_R,\hat{v}\mapsto x_L\}$ is the pinning function. We observe that the vertices in~$\lpart$ can only be mapped to 
vertices in~$\{x_L,u_1,u_2,u_3\}$ and vertices in~$\rpart$ can only be mapped to vertices in~$\{x_R, v_1, v_2\}$. This 
observation yields that the number of homomorphisms from~$J$ to~$H$ is equivalent to~$\sum_{I\in\calI(G)}3^{|\lpart \cap 
I|}2^{|\rpart \cap I|} \pmod p$. Furthermore, the cardinality of the sets $\{u_1,u_2,u_3\}$ and $\{v_1,v_2\}$ 
introduces weights in a natural way.

For the reduction above we need the following property easily observable in $H$: there exist two adjacent 
vertices of degree $a=\lweight+1\not\equiv 1\pmod p$ and $b=\rweight+1\not \equiv 1 \pmod p$. Recall that in order to 
obtain hardness for \pbislr{\lweight}{\rweight} Theorem~\ref{thm:wbis-hardness} requires $\lweight,\rweight\not\equiv 
0\pmod p$. In fact, as we will show in Section~\ref{sec:trees}, these vertices need not be adjacent. During the 
construction of $J$ we can replace the edges of $G$ with paths of appropriate length. We call such a structure in $H$ 
an $(a,b,p)$-path. In Lemma~\ref{lem:Homs_trees_hardcase} we formally prove that if $H$ has an $(a,b,p)$-path, then 
\phcol{} is $\numpp$-hard. In particular, observe that stars cannot contain $(a,b,p)$-paths. Finally, we show that every non-star tree~$H$ 
contains an $(a,b,p)$-path, which yields our main result on \phcol{} (Lemma~\ref{lem:tree_star_or_hard}).

\subsection{Composites}
\label{sec:intro-composites}

We outline the obstacles occurring when extending the dichotomy for \khcol{} to any integer~$k$. Let 
$H$ be a graph and let $k= \prod_{i=1}^m k_i$, where $k_i= p_i^{r_i}$ is an integer with its prime factorisation. 
Assuming \khcol{} can be 
solved in polynomial time, then for each $i\in[m]$,  $\#_{k_i}\prb{HomsTo}H$ can also be solved in polynomial 
time. The reason is that $k_i$ is a factor of $k$ and we can apply the modulo~$k_i$ operator to the answer 
for the \khcol{} instance. The Chinese remainder theorem shows that the converse is also true: if for each $i\in[m]$ we 
can solve $\#_{k_i}\prb{HomsTo}H$ in polynomial time, then we can also solve \khcol{} in polynomial time. By the previous observations we can now focus 
on powers of primes~$k=p^r$. Assuming \khcol{} is computable in polynomial time yields again that \phcol{} is also 
computable in polynomial time. However, the converse is not always true. 

Guo et al.~\cite{GHLX11} were able to obtain this reverse implication for the constraint satisfaction problem. They 
showed~\cite[Lemma~4.1 and Lemma~4.3]{GHLX11} that when $p$ is a prime $\kcsp[{p^r}]$ is computable in polynomial time if \pcsp{} is computable in polynomial time. In Section~\ref{sec:composites} we show that their technique cannot be transferred to the~\khcol{} setting. We show that 
there is a graph~($P_4$) such that $\#_2\prb{HomsTo}P_4$ is computable in polynomial time, while $\#_4\prb{HomsTo}P_4$ 
is $\parp$-hard.

\subsection{Organisation}
\label{sec:organisation}
Our notation is introduced in Section~\ref{sec:prelim}.
In Section~\ref{sec:bipartite} we study the complexity of the weighted bipartite independent sets problem modulo any prime.
Section~\ref{sec:poly} presents the connection to the polynomial time algorithm of Faben and Jerrum for \phcol{}.
Our pinning method is explained in Section~\ref{sec:pinning}.
Section~\ref{sec:trees} contains the hardness reduction for \phcol{}.
Our results are collected into a dichotomy theorem in Section~\ref{sec:main-thm}. Finally, in Section~\ref{sec:composites} we discuss the obstacles arising when counting modulo all integers.

%% file: preliminaries.tex
We denote by $[n]$ the set $\{1, \dots, n\}$.
Further, if $v$ is an element of the set~$S$, we write $S-v$ for $S\setminus
\{v\}$.
Let $k$ be a positive integer $k\in\Z_{>0}$, then for a function $f$ its $k$-fold composition is denoted by $f^{(k)}=f\circ f\circ\dots\circ f$.

\bigskip
For a detailed introduction to Graph Theory the reader is referred to \cite{west_introduction_2000}. \vspace{-1.5ex}
\paragraph{(Simple) graphs}
Unless otherwise specified, \emph{graph}s are undirected and simple, requiring them to contain neither parallel edges nor loops.
More formally, a graph $G$ is a pair $(V,E)$, where $V$ denotes the set of vertices and $E\subseteq V\times V$ the set of edges formed by pairs of vertices. This set of edges can be looked upon as a relation for a pair of vertices to either form an edge or not in $G$. For a graph $G$ we sometimes denote its vertex set by $V(G)$ and its edge set by $E(G)$.
As stated above, for all vertices $u,v\in V$ we require the edges to be undirected, that is $(u,v)\in E$ if and only if $(v,u)\in E$ Moreover, for an edge $(u,v)\in E$ the condition $u\neq v$ ensures the absence of loops. A graph $H$ is a \emph{subgraph} of $G$ if $V(H) \subseteq V(G)$ and $E(H) \subseteq E(G)$. This is denoted by $H \subseteq G$. If additionally $(u,v) \in E(G)$ such that $u,v \in V(H)$ implies that $(u,v) \in E(H)$ we call $H$ an \emph{induced subgraph}. In fact, $H$ is then induced by the subset $V(H) \subseteq V(G)$. 
For all vertices $v \in V(G)$ of a graph $G$ with a subgraph $H$ we denote by $\Gamma_H(v) = \{\,u\in V\mid (u,v)\in E\,\}$ the \emph{neighbourhood of $v$ in $H$} containing all vertices in $V(H)$ adjacent to $v$, which refers to the Greek term \greektext Geitoniá\latintext. Consequently, we denote by $\deg_H(v)$ the size of $\Gamma_H(v)$.

A \emph{path} is a simple graph $P$ such that all its vertices can be ordered in a list without multiples and only two adjacent vertices in the list form an edge in $E(P)$. The \emph{length of a path} is its number of edges. 
Two vertices $u, v$ in a graph $G$ are \emph{connected} if there exists a path $P \subseteq G$, such that $u,v \in V(P)$. Otherwise the vertices are \emph{disconnected}. If every pair of vertices in a graph $G$ is connected, then $G$ is called connected. Otherwise it is called disconnected. We call a subgraph $H \subseteq G$ a \emph{connected component} of $G$, if $H$ is connected and there exists no vertex $v \in V(G) \setminus V(H)$ such that $v$ is connected to any vertex in $H$. An \emph{independent set} of a graph $G$ is a set of vertices $I\subseteq V(G)$, such that no pair of vertices in $I$ is connected in $G$.
The \emph{distance of two connected vertices} $u,v$ in $G$, denoted by $d_G(u,v)$, is the length of a shortest path in $G$ connecting $u$ and $v$. A \emph{cycle} is a simple connected graph $C$, such that all its vertices can be ordered in a list $v_0 v_1\dots v_k v_0$ and only two vertices adjacent in the list form an edge. A \emph{tree} is a simple connected graph, which does not contain cycles.
For an integer $k \geq 0$ a $k$-\emph{walk} is a list $v_0 v_1 \dots v_k$, which might contain multiples, such that two adjacent vertices in the list form an edge.

A graph $G$ is \emph{bipartite} if there exist disjoint subsets $\lpart$, $\rpart$ of $V$ such that $V=\lpart \cup \rpart$ and there exists no edge $(u,v) \in E$ with $u,v \in \lpart$ or $u,v \in \rpart $. 
We write $G= \bipG$ for the bipartite graph with fixed components $\lpart$ and $\rpart$, which we are calling the \emph{left and right component}, respectively.
For an integer $k\geq 0$ the \emph{complete graph} of size $k$ is the simple graph $K_k$ with $|V(K_k)| =k$ and every pair of distinct vertices in $V(K_k)$ forms an edge. Similarly, for integers $k_\L, k_\R \geq 0$ the \emph{complete bipartite graph} is the simple bipartite graph $K_{k_\L, k_\R}= \bipG$ with $|\lpart|=k_\L$ and $|\rpart|=k_\R$ and every pair of vertices $u \in \lpart$, $v\in \rpart$ forms an edge.
A \emph{star} is a complete bipartite graph $K_{1,k}$ for some integer $k\geq 0$. 

Let $G$ and $H$ be graphs.
A \emph{homomorphism from $G$ to $H$} is a function $\sigma:V(G)\to V(H)$, such that edges are preserved, for short $(v_1,v_2)\in E(G)$ implies $(\sigma(v_1),\sigma(v_2))\in E(H)$.
Moreover, $\Homs{G}{H}$ denotes the set of homomorphisms from $G$ to $H$.
An \emph{isomorphism between $G$ and $H$} is a bijective function $\varrho:V(G)\to V(H)$ preserving the edge relation in both directions, meaning $(v_1,v_2)\in E(G)$ if and only if $(\varrho(v_1),\varrho(v_2))\in E(H)$.
If such an isomorphism exists, we say that $G$ is isomorphic to $H$ and denote it with $G\isoto H$. 
An \emph{automorphism of $G$} is an isomorphism from the graph $G$ to itself. $\Aut(G)$ denotes the \emph{automorphism group of $G$}.
An automorphism $\varrho$ is an \emph{automorphism of order~$k$} in case it is not the identity and $k$ is the smallest positive integer such that $\varrho^{(k)}$ is the identity.

\paragraph{Partially labelled graphs}
Let $H$ be a graph.
A \emph{partially $H$-labelled graph} $J=(G,\tau)$ consists of an \emph{underlying graph} $G(J)=G$ and a (partial) \emph{pinning function} $\tau(J)=\tau: V(G)\to V(H)$, mapping vertices in $G$ to vertices in $H$.
Every vertex~$v$ in the domain $\dom(\tau)$ of $\tau$ is said to be \emph{$H$-pinned to $\tau(v)$}.
We omit $H$ in case it is immediate from the context.
We denote a partial function $\tau$ with finite domain $\{v_1,\dots,v_r\}$ also in the form $\tau = \{v_1\mapsto \tau(v_1),\dots,v_r\mapsto \tau(v_r)\}$.
A \emph{homomorphism from a partially labelled graph~$J$ to a graph~$H$} is a homomorphism from $G(J)$ to~$H$ that respects~$\tau$, i.e., for all $v\in\dom(\tau)$ holds $\sigma(v)=\tau(v)$.
By $\Homs{J}{H}$ we denote the set of homomorphisms from $J$ to $H$.

\paragraph{Graphs with distinguished vertices}  
Let $G$ and $H$ be a graphs.
It is often convenient to regard a graph with a number of (not necessarily distinct) distinguished vertices $v_1, \dots, v_r$, which we denote by $(G, v_1, \dots, v_r)$.
A \emph{sequence of vertices} $v_1\ldots v_r$ may be abbreviated by $\vbar$ and
$G[\vbar]$ stands for the subgraph of~$G$ induced by the set of vertices $\{v_1,\ldots,v_r\}$.
A \emph{homomorphism from $(G,\ubar)$ to $(H,\vbar)$} with $r=|\ubar|=|\vbar|$ is a homomorphism~$\sigma$ from $G$ to~$H$ with $\sigma(u_i)=v_i$ for each $i\in[r]$.
Such a homomorphism immediately yields a homomorphism from the partially labelled graph $(G, \{u_1\mapsto v_1, \dots, u_r\mapsto v_r\})$ to~$H$ and vice versa.
For a partially labelled graph~$J$ and vertices $u_1,\dots,u_r\notin
\dom(\tau(J))$, we identify a homomorphism from $(J, \ubar)$ to $(H,\vbar)$ with the corresponding homomorphism from $(G(J), \tau(J) \cup \{u_1\mapsto v_1, \dots, u_r\mapsto v_r\})$ to~$H$.
Similarly, $(G, \ubar)$ and $(H, \vbar)$ are \emph{isomorphic} if $r=|\ubar|=|\vbar|$ and there is an isomorphism $\varrho$ from $G$ to $H$, such that $\varrho(u_i)=v_i$ for each $i\in[r]$.
An \emph{automorphism of $(G, \ubar)$} is an automorphism~$\varrho$ of~$G$ with the property that $\varrho(u_i)=v_i$ for each $i\in[r]$ and $\Aut(G,\ubar)$ denotes the \emph{automorphism group of $(G,\ubar)$}.

\paragraph{Reductions}
For a detailed discussion of this topic see~\cite{PapCC}.
Our model of computation is the standard multitape Turing machine.
For counting problems $P$ and $Q$, we say that \emph{$P$ reduces to $Q$ via polynomial time Turing reduction}, if there is a polynomial time deterministic oracle Turing machine $M$ such that, on every instance $x$ of $P$, $M$ outputs $P(x)$ by making queries to oracle $Q$.
Further, \emph{$P$ reduces to $Q$ via parsimonious reduction}, if there exists a polynomial time computable function $f$ transforming every instance~$x$ of $P$ to an instance of $Q$, such that $P(x)=Q(f(x))$.
Clearly, if $P$ reduces to $Q$ via parsimonious reduction, then $P$ also reduces to $Q$ via polynomial time Turing reduction.

\paragraph{Basic algebra}
For an introduction to abstract algebra we refer the reader to~\cite{AlgebraBook}.
Finally, we assume familiarity with the notion of a \emph{group}, an \emph{action of a group on a set} and modular 
arithmetic in the \emph{field $\Z_p$}, where $p$ is a \emph{prime} in $\Z$. 
We are going to apply Fermat's little theorem (see~\cite[Theorem~11.6]{Arm88:GroupSym}) and Cauchy's group theorem (see, e.g., \cite[Theorem~13.1]{Arm88:GroupSym}) frequently.

\begin{theorem}[Fermat's little theorem]\label{thm:Fermat}
Let $p$ be a prime. If $a \in \Z$ is not a multiple of $p$, then $a^{p-1}\equiv 1\pmod p$.
\end{theorem}

\begin{theorem}[Cauchy's group theorem]
\label{thm:Cauchy}
Let $p$ be prime.
If $\calG$ is a finite group and $p$ divides~$|\calG|$, then $\calG$~contains an element of order~$p$.
\end{theorem}

%% file: bipartite.tex
We study the complexity of computing the weighted sum over independent sets in a bipartite graph modulo a prime.
This weighted sum is denoted by \pbislr{\lweight}{\rweight}, where $\lweight,\rweight\in\Z_p$ are weights the 
vertices of each partition contribute. For this, the input bipartite graphs come with a fixed partitioning of their 
vertices. Note that the set of independent sets of a graph does not change if the graph contains multiedges and further 
note that a bipartite graph cannot contain loops. For this reason, in this section we do not have to 
distinguish between a bipartite multigraph or a bipartite simple graph.



\medskip
For a graph $G$ let $\calI(G)$ denote the set of \emph{independent sets of $G$}. Faben~\cite[Theorem~3.7]{Fab08} shows 
that the problem $\kbis{k}$ of counting the independent sets of a graph modulo an integer~$k$ is hard, even when the 
input graph is restricted to be bipartite.


\begin{theorem}[Faben]\label{thm:bis}
	For all positive integers~$k$, \kbis{k} is $\nkp$-complete.
\end{theorem}

Let $p$ be a prime and $\lweight,\rweight\in\zp$, we will 
study the complexity of computing the following weighed sum over independent sets of a bipartite graph $G=\bipG$ modulo~$p$
\[
	\wISet{G}=\sum_{I\in\calI(G)}\lweight^{|\lpart \cap I|}\rweight^{|\rpart \cap I|}.
\]

We note that every bipartite graph contains a partition $\lpart,\rpart$ and declaring a bipartite graph with $G=\bipG$ is the same as having the graph $G$ along with the partition as input.
A given partition is necessary when studying weighted independent, since changing the partitioning changes the value of the weighted sum.
In the unweighted sum of Theorem~\ref{thm:bis}, there is no need to give a fixed partition as input, as it does not change 
the number of independent sets.
More formally, we study the following problem.

{\renewcommand{\thetheorem}{\ref{prob:bislr}}
\probpar{\pbislr{\lweight}{\rweight}.}
	{$p$ prime and $\lweight, \rweight \in\Z_p$.}
	{Bipartite graph $G=\bipG$.}
	{$\wISet{G} \pmod p$.}}

As a note, $\pbislr{1}{1}$ corresponds to the special case $\kbis{p}$. Theorem~\ref{thm:bis} directly implies that \pbislr{1}{1} is $\numpp$-complete for all primes~$p$.

\medskip
We begin by identifying the tractable instances of $\pbislr{\lweight}{\rweight}$.

\begin{proposition}\label{prop:BIS_easy_cases}
	If $\lweight\equiv 0 \pmod p$  or $\rweight \equiv 0 \pmod p$ then \pbislr{\lweight}{\rweight} is computable in polynomial time. 
\end{proposition}
\begin{proof}
	Without loss of generality we assume $\lweight \equiv 0 \pmod p$. Thus, any 
	independent set that contains at least one vertex from $\lpart$ contributes zero to the sum in $\wISet{G}$. 
	Therefore, we only need to consider the independent sets~$I$ with $I\not\subseteq\lpart$. Since any subset of $\rpart$ yields an independent set, we obtain
	\begin{align*}
		\wISet{G} & \equiv 1+\sum_{i=1}^{|\rpart|}\binom{|\rpart|}{i}(\rweight)^i\pmod p \\
		 & = \sum_{i=0}^{|\rpart|}\binom{|\rpart|}{i}(\rweight)^i 
		 = (\rweight+1)^{|\rpart|},
	\end{align*}
	which can be computed in polynomial time.
\end{proof}

The remainder of the section is dedicated to proving that \pbislr{\lweight}{\rweight} is hard in all other cases.
Our reduction is inspired by the reduction of Faben in~\cite[Theorem~3.7]{Fab08}.  

In the proofs that follow, to avoid double counting, it is useful to partition the independent sets in the 
following way.

\begin{definition}\label{def:is-l-r-lr}
Let $G=\bipG$ be a bipartite graph.
We denote by $\calI{_L}(G)$ the set $\{I\in\calI(G)\setminus\{\varnothing\}\mid I\subseteq V_L \}$ of non-empty independent sets containing only vertices from $\lpart$.
Similarly, we write $\calI{_\R} (G)$ for the set of non-empty independent sets that contain only vertices from $\rpart$.
Finally, we denote by $\calI_{\L \R} (G)$ the set $\calI(G)\setminus(\,\calI_L(G)\cup\calI_R(G)\cup\{\varnothing\}\,)$ of independent sets containing at least one vertex in $\lpart$ and at least one vertex in $\rpart$.
\end{definition}

%
Given a bipartite graph $G$, the following lemma expresses $\wISet{G}$ in terms of the partitioning defined above.

\begin{lemma}\label{lem:split-sum-part}
Let $G=\bipG$ be a bipartite graph.
Then,
\[
	\wISet{G} =(\lweight+1)^{|\lpart|}+(\rweight+1)^{|\rpart|}-1+\sum_{I\in\calI_{\L \R}(G)} \lweight^{|\lpart\cap I|} \rweight^{|\rpart \cap I|}.
\]
\end{lemma}
\begin{proof}
By Definition~\ref{def:is-l-r-lr} the set $\calI(G)$ partitions into $\{\calI_\L (G),\calI_\R (G),\calI_{\L\R}(G),\{\emptyset\}\}$, which yields
\begin{align}
		\wISet{G}& =\sum_{I\in\calI(G)}\lweight^{|\lpart\cap I|}\rweight^{|\rpart\cap I|} \nonumber \\
		 &= \sum_{I\in\calI_\L (G)}\lweight^{|I|}+\sum_{I\in\calI_\R(G)}\rweight^{|I|}+\sum_{I\in\calI_{\L\R}(G)}\lweight^{|\lpart \cap I|}\rweight^{|\rpart \cap I|}+1. \label{eq:tot-sum}
\end{align}

As in the proof of Proposition \ref{prop:BIS_easy_cases}, we obtain
\begin{align}
	 \label{eq:left-is}
	 \sum_{I\in\calI_\L(G)} \lweight^{|I|} &= \sum_{i=0}^{|\lpart|}\binom{|\lpart|}{i}\lweight^i - 1 = (\lweight+1)^{|\lpart|}-1, \qquad \mbox{and analoguously} \\
	 \label{eq:right-is}
	 \sum_{I\in\calI_\R (G)} \rweight^{|I|} &= (\rweight+1)^{|\rpart|}-1.
\end{align}
Inserting \eqref{eq:right-is} and \eqref{eq:left-is} into \eqref{eq:tot-sum} yields the lemma.
\end{proof}

For our reduction to work though, we must design gadgets which are tailored to our general setting of weighted 
independent sets.

\begin{definition}\label{def:special_bip_graph}
Let $p$ be a prime and let $\lweight, \rweight \in \zsp$.

For every $k \in [p]$ we denote by $B(k, p)=\bipG$ the bipartite graph  with $4(p-1)$ vertices in two disjoint vertex sets $V_L := \{u_1,\dots,u_{2(p-1)}\}$, $V_R := \{v_1,\dots,v_{2(p-1)}\}$ and the edge set
$$E:=\{(u_i,v_j)\mid i,j\in[2(p-1)], \textrm{ where } i\neq j\}\cup\{(u_i,v_i)\mid i \notin [k]\},$$
consisting of all edges in the complete bipartite graph $K_{2(p-1), 2(p-1)}$ except $(u_i, v_i)$ with $i\in [k]$.	 
\end{definition}

See Figure \ref{fig:bip_gadget} for the example graph $B(1, 3)$.

\begin{figure}[t]
\centering
\includegraphics{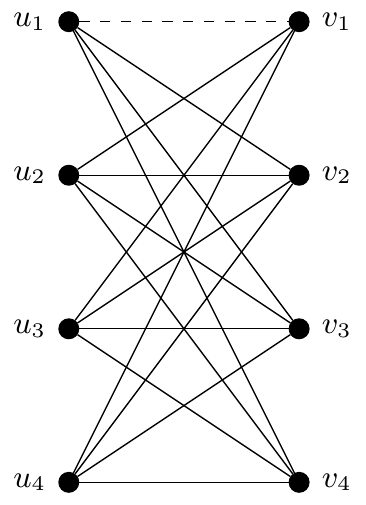}
\caption{Constructive route for $p=3$ and $k=1$. Starting with the complete bipartite graph $K_{4,4}$ the edge $(u_1, v_1)$ is removed.}
\label{fig:bip_gadget}
\end{figure}

\medskip
$B(k,p)$ has two types of vertices in each partition:
the vertices in $\{u_i,v_i\}_{i\leq k}$ of degree $2(p-1)-1$ and the vertices in $\{u_i,v_i\}_{i>k}$ of degree $2(p-1)$.
Since the size of the vertex sets is a multiple of $(p-1)$ we are able to apply Fermat's little Theorem~\ref{thm:Fermat} in our reductions later on. Moreover, the size is large enough to generate every necessary value of $k\in [p]$.
This freedom of choice for $k$ will entail the possibility to, given $\lweight, \rweight \not\equiv 0\pmod p$, choose $k$ such that $\wISet{B(k, p)} \equiv 0 \pmod p$. Given such a $k$, we will see that in each partition there exists a vertex $v$ such that removing this vertex from $B(k,p)$ will yield $\wISet{B(k, p) - v} \not\equiv 0 \pmod p$. This property will be crucial later on.

\medskip
The following lemma establishes the key properties of the bipartite~$B(k,p)$ defined above and will be later used to 
show that our reduction gadgets behave as we want.

\begin{lemma}\label{lem:sum-is-lr}
	Let $p$ be a prime, $\lweight,\rweight\in\zsp$, $k\in \Z_p$ and $B=B(k,p)$ as in Definition~\ref{def:special_bip_graph}. Then, 
	\[
		\sum_{I\in\calI_{\L\R}(B)}\lweight^{|\lpart\cap 
		I|}\rweight^{|\rpart\cap I|} \equiv k \lweight\rweight\pmod p.
	\]
\end{lemma}
\begin{proof}
	Let $I\in\calI_{\L\R}(B)$ be a non-empty 
	independent set containing a vertex $u_i \in \rpart$ and a vertex $v_j \in \lpart$.
	By the definition of $B$ there is no independent set containing two vertices $u_i$ and $v_j$ with $i \neq j$.
	Thus $i=j$ and $V_L\cap I=\{u_i\}$ as well as $V_R\cap I=\{v_i\}$.
	We obtain $I = \{u_i, v_i\}$ yielding $\calI_{LR}=\{\{u_i,v_i\}\mid i\in[k]\}$.
\end{proof}

The following lemma states the properties of the graphs we will use as gadgets, namely a copy of a $B(k,p)$ for an appropriately chosen~$k\in[p]$, together with two distinguished vertices.

\begin{lemma}\label{lem:bip-gadget}
	Let~$p$ be a prime and $\lweight, \rweight \in \zsp$. There exists a bipartite graph $B= \bipG$ with 
	distinguished vertices $u_\L\in \lpart$ and~$v_\R\in \rpart$, that satisfies
	\begin{enumerate}
		\item $\wISet{B} \equiv 0\pmod p$,
		\item $\wISet{B-u_\L} \not\equiv 0 \pmod p$,
		\item $\wISet{B-v_\R} \not\equiv 0 \pmod p$.
	\end{enumerate}
\end{lemma}
\begin{proof}
	As pointed out, the family of graphs $B(k,p)$ contains at least one graph with the desired properties given the weights $\lweight, \rweight \in \zsp$. For every graph $B=B(k,p)$ we apply Lemma~\ref{lem:split-sum-part} to obtain
	\begin{align}\label{eq:wISet_B}
	\nonumber
	\wISet{B}&=(\lweight+1)^{|\lpart|}+(\rweight+1)^{|\rpart|}-1+\sum_{I\in\calI_{\L\R}}\lweight^{|\lpart\cap I_{\L\R}|}\rweight^{|\rpart\cap	I_{\L\R}|}\\
	&= (\lweight+1)^{2(p-1)}+(\rweight+1)^{2(p-1)}-1+\sum_{I\in\calI_{\L\R}}\lweight^{|\lpart\cap I_{\L\R}|}\rweight^{|\rpart\cap	I_{\L\R}|}.
	\end{align}
	If one of the weights is equivalent to $-1$ in $\zp$ the corresponding term in \eqref{eq:wISet_B} vanishes. Otherwise, we are allowed to apply Fermat's little Theorem~\ref{thm:Fermat} and the corresponding term is equivalent to $1$.
	Therefore, we have to distinguish cases.

	\begin{enumerate}[wide, labelwidth=!, labelindent=0pt, parsep=0pt]
		\item[i. $\lweight, \rweight \not \equiv-1\pmod p$.] $\mbox{}$\newline
		 	We can apply Fermat's little Theorem~\ref{thm:Fermat} on the terms corresponding to both weights. In conjunction with Lemma~\ref{lem:sum-is-lr} this yields
		 	\begin{align*}
		 		\wISet{B}&\equiv 1 + k\lweight\rweight\pmod p .
		 	\end{align*}
		 	Now we choose $k \in \zsp$ satisfying $k \equiv - (\lweight\rweight)^{-1} \pmod p$ and property~$\mathit{1}$ follows. We note that such a $k$ uniquely exists since $p$ is a prime and $\zp$ a field.
		 	In order to prove the remaining properties, we choose $u_\L = u_{2(p-1)}$ and $v_\L = v_{2(p-1)}$. We observe that removing any of these two vertices from $V(B)$ does not affect the independent sets in $\calI_{\L \R}$. The reason is that in $B$ the vertices $u_\L$ and $v_\R$ are connected to every vertex in $\lpart$ and $\rpart$, respectively. We derive due to the choice of $k$ by \eqref{eq:wISet_B}
		 	\begin{align*}
			 	\wISet{B- u_\L} &\equiv (\lweight + 1)^{2(p-1)-1} - 1 \pmod p \equiv (\lweight + 1)^{-1} - 1 \pmod p; \\
			 	\wISet{B - v_\R} &\equiv (\rweight + 1)^{2(p-1)-1} - 1 \pmod p \equiv (\rweight + 1)^{-1} - 1 \pmod p.
		 	\end{align*}
		 	We note that both expressions are not equivalent to $0$ in $\zp$ since both weights are in $\zsp$.
		 \item[ii. $\lweight \equiv -1 \pmod p ,\, \rweight \not \equiv -1 \pmod p$.] $\mbox{}$\newline
		 	Lemma~\ref{lem:sum-is-lr} in conjunction with Fermat's little Theorem~\ref{thm:Fermat} on the term corresponding to the weight $\rweight$ yields
		 	\begin{align*}
		 		\wISet{B}&\equiv k\lweight\rweight\pmod p .
		 	\end{align*}
		 	We note that the definition of $B(k,p)$ also allows us to choose $k=p$, which we are doing in this case. This choice proves property~$\mathit{1}$. However, this implies that we cannot choose the same vertices as we did in the first case to prove the remaining properties. In particular, we have to adjust the choice for the vertex $u_\L$ corresponding to the weight $\lweight$.
		 	
		 	Regarding the vertex in $\rpart$, we choose again $v_\R = v_{2(p-1)}$. Similar to the observation in the first case this yields
		 	\begin{align*}
		 		\wISet{B - v_\R} &\equiv (\rweight + 1)^{2(p-1)-1} - 1 \pmod p \equiv (\rweight + 1)^{-1} - 1 \pmod p.
		 	\end{align*}
		 	Regarding the vertex in $\lpart$, we choose $u_\L = u_k$. We note that the edge $(u_k, v_k)$ is missing in $B$ and therefore the set $\{u_k, v_k\}$ is in $\calI_{\L \R} (B)$. Therefore, for the choice of $u_\R$ the set $\{u_k, v_k\}$ cannot be in $\calI_{\L \R} (B - u_\L)$. In particular, we obtain $\calI_{\L \R} (B - u_\L) = \calI_{\L\R}(B) - \{u_k, v_k\}$. We deduce
		 	\begin{align*}
		 		\wISet{B - u_\L} &\equiv \sum_{I\in\calI_{\L\R}(B-u_\L)}\lweight^{|\lpart\cap I_{\L\R}|}\rweight^{|\rpart\cap	I_{\L\R}|} = (k - 1) \lweight \rweight.
		 	\end{align*}
		 	Due to the choice of $k=p$ in conjunction with $\lweight \equiv -1$ this simplifies to the desired
		 	\begin{equation*}
			 		\wISet{B - u_\L} \equiv \rweight,
		 	\end{equation*}
		 	which cannot be equivalent to $0$ since $\rweight \in \zsp$.
		 \item[iii. $\lweight \not\equiv -1 ,\, \rweight\equiv-1\pmod p$.] $\mbox{}$\newline
		 	The proof of this case in analogue to the second case. For this purpose we need to interchange the role of the left and right partition. In particular, choosing $k=p$ as well as $u_\L = u_{2(p-1)}$ and $v_\R = v_k$ establishes this case.
		 \item[iv. $\lweight, \rweight \equiv-1\pmod p$.] $\mbox{}$\newline
		 	This case will be proven with a variation of the arguments used in the above cases. Since both weights are such that the corresponding terms in \eqref{eq:wISet_B} are vanishing, we obtain
		 	\begin{align*}
		 		\wISet{B}&\equiv - 1 + k\lweight\rweight\pmod p .
		 	\end{align*}
		 	In fact, this is almost the same situation we faced in the first case. We choose $k \in \zp$ satisfying $k\equiv (\lweight\rweight)^{-1} \pmod p$ yielding the first property. Due to this case's assumption this implies $k=1$.
		 	Similar to the situation faced in the second and third case we have to choose $u_\L$ and $v_\R$ such that the removal of one of these vertices from $B$ affects the independent sets in $\calI_{\L \R}$. Therefore, we choose again $u_\L = u_k$ and $v_\R = v_k$. This choice has the same effect on $\calI_{\L \R}$ as we have observed above. We deduce $\sum_{I\in\calI_{\L\R}(B-u_\L)}\lweight^{|\lpart\cap I_{\L\R}|}\rweight^{|\rpart\cap	I_{\L\R}|} = (k-1) \lweight \rweight = 0$ and thus
		 	\begin{align*}
			 	\wISet{B - u_\L} & = -1 + \sum_{I\in\calI_{\L\R}(B-u_\L)}\lweight^{|\lpart\cap 
I_{\L\R}|}\rweight^{|\rpart\cap	I_{\L\R}|} = -1 , \quad \mbox{and analoguously} \\
			 	\wISet{B - v_\R} &= -1.
		 	\end{align*}
		 	This establishes the lemma. \qedhere
	\end{enumerate}

\end{proof}

\medskip

\medskip
As in the proof of~\ref{thm:bis} we use $\ksat{}$ as a starting problem. 
Given a Boolean formula~$\varphi$, let $\sat(\varphi)$ be the set of the satisfying assignments of~$\varphi$.

\probpar{\ksat{}.}
  {$k$ integer.}
  {Boolean formula~$\varphi$ in conjunctive normal form.}
  {$|\sat(\varphi)|\pmod k$.}

Simon in his thesis~\cite[Theorem 4.1]{Sim75} shows how the original reduction of Cook can be made parsimonious. 
As Faben observes in~\cite[Theorem~3.1.17]{Fabenthesis} any parsimonious reduction is parsimonious modulo~$k$, for any integer~$k$, hence \ksat{} is $\nkp$-complete.

\begin{theorem}[Simon]\label{thm:sat-hardness}
\ksat{} is $\nkp$-complete under parsimonious reductions for all integers $k$.
\end{theorem}

Let $p$ be a prime. Our reduction starts from a Boolean formula $\varphi$, input for \numpsat{}, and constructs, in two 
stages, a graph $G_\varphi$, input for \pbislr{\lweight}{\rweight}.

In the first stage we define the graph $G'_\varphi$. For every variable $x_i$ in $\varphi$, 
$G'_\varphi$ contains three vertices $u_i$, $\ubar_i$ and $w_i$ to the left vertex set $\lpart(G'_\varphi)$ as well as 
three vertices $v_i$, $\vbar_i$ and $z_i$ to the right vertex set $\rpart (G'_\varphi)$. For every clause 
$c_j$ of $\varphi$, $G'_\varphi$ further contains a vertex $y_j$ in the right vertex set $\rpart(G'_\varphi)$.
We further introduce the edges forming the cycle $u_iv_iw_i\vbar_i\ubar_iz_iu_i$ to $E(G'_\varphi)$ for every variable $x_i$ in $\varphi$.
Additionally for all $i\in[n]$, if $x_i$ appears as a literal in clause~$c_j$ of $\varphi$, we introduce the edge 
$(u_i,y_j)$ in $G'_\varphi$ and if $\bar{x_i}$ appears as a literal in clause~$c_j$, we introduce the edge 
$(\ubar_i,y_j)$ in $G'_\varphi$. The left part of Figure~\ref{fig:weighted_BIS} illustrates an example of this 
construction. Formally, $G'_\varphi$ is defined as follows.

\begin{figure}[t]
\centering
\includegraphics[width=\textwidth]{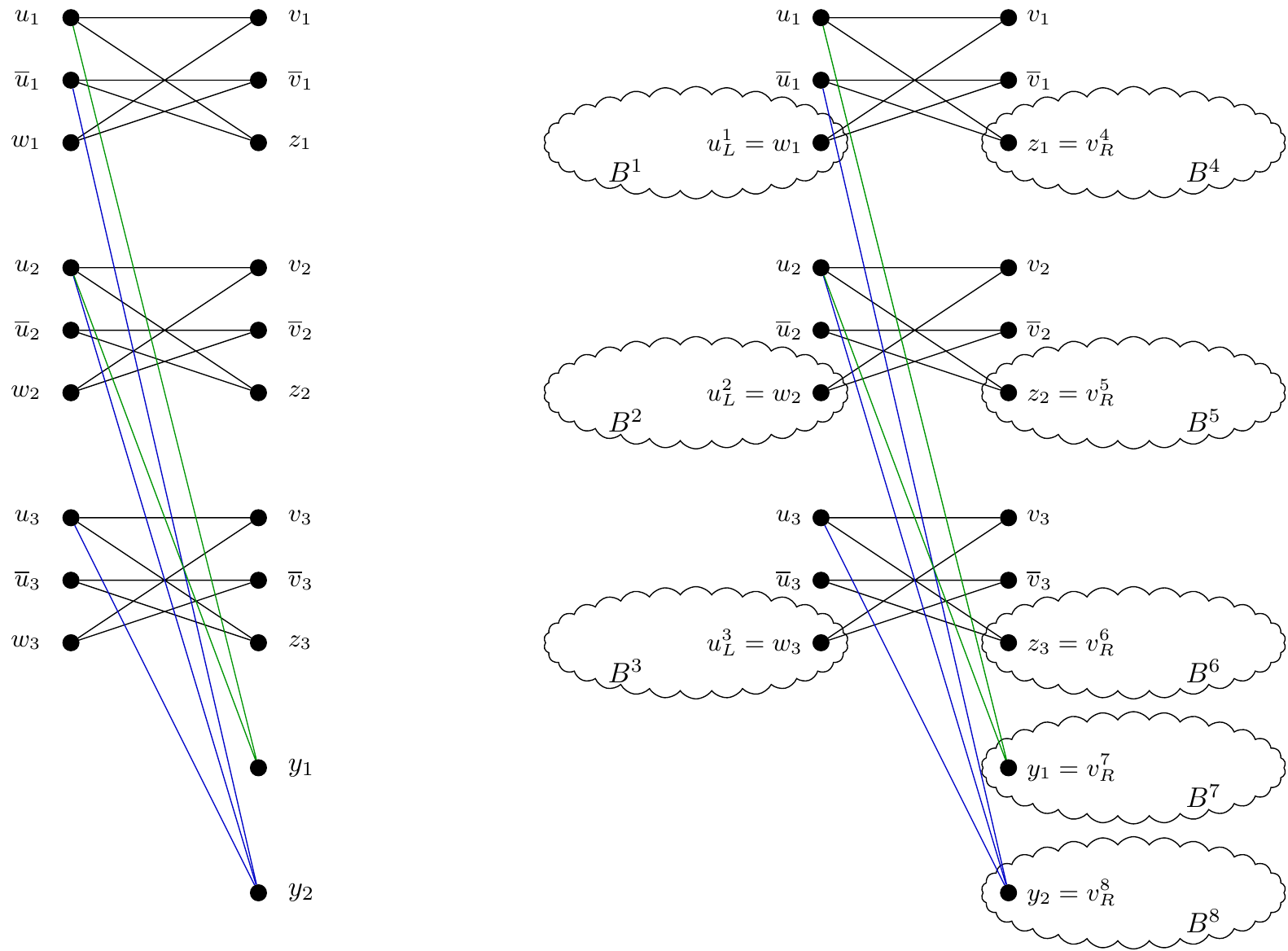}
\caption{The graphs $G'_\varphi$ and $G_\varphi$ for $\varphi = \textcolor{mycolor2}{(x_1 \vee x_2)} \wedge 
\textcolor{mycolor1}{(\bar{x}_1 \vee x_2 \vee x_3)}.$}\label{fig:weighted_BIS}
\end{figure}

\begin{definition} \label{def:auxgraph}
Let $\varphi$ be a Boolean formula in conjunctive normal form with variables $x_1,\ldots,x_n$ and clauses $c_1,\ldots,c_m$.
The bipartite graph $G'_\varphi = (\lpart (G'_\varphi), \rpart (G'_\varphi), E(G'_\varphi))$ is defined by
\begin{align*}
\lpart(G'_\varphi)=&\{\,u_i,\ubar_i,w_i\mid i\in[n]\,\}, \\
\rpart(G'_\varphi)=&\{\,v_i,\vbar_i,z_i\mid i\in[n]\,\}\cup\{\,y_j\mid j\in[m]\,\} \text{ and} \\
E(G'_\varphi)=&\{\,(u_i,v_i),(v_i, w_i), (w_i, \bar{v_i}), (\bar{v_i},\bar{u_i}),(\bar u_i,z_i),(z_i,u_i) \mid i\in[n] \,\}  \\
&\cup \{\,(u_i, y_j) \mid i\in[n], \,  j\in[m]  \text{ and }  x_i \text{ occurs in } c_j\,\} \; \\
&\cup \{\,(\ubar_i, y_j) \mid i\in[n] , \, j\in[m] \text{ and } \xbar_i \text{ occurs in } c_j\,\}.
\end{align*}
\end{definition}
%
%
%
Note that $G'_\varphi$ is bipartite, since there are no adjacent vertices in both partition sets.

\smallskip
In the second and final stage, we construct the graph~$G_\varphi$. Let $(B,u_\L,v_\R)$ be the graph obtained from Lemma~\ref{lem:bip-gadget}.
$G_\varphi$ is a copy of $G'_\varphi$ together with two copies of $B$ for every variable of $\varphi$ and one copy of $B$ for every clause.
The first $n$ copies $B^1,\ldots,B^n$ are connected to $G'_\varphi$ by identifying the distinguished vertex $u^i_L$ in the left component with $w_i \in \lpart (G'_\varphi)$ for all $i \in [n]$.
The second $n$ copies $B^{n+1},\ldots,B^{2n}$ are connected to $G'_\varphi$ by identifying the distinguished vertex $v^{n+i}_R$ in their right components with $z_i 
\in \rpart (G'_\varphi)$ for all $i \in [n]$.
The remaining $m$ copies $B^{2n+1},\ldots,B^{2n+m}$ of $B$ are connected to $G'_\varphi$ by identifying the distinguished vertex $v_R^{2n+j}$ in their right components with $y_j \in \rpart (G'_\varphi)$ for all $j \in [m]$.
For an example see the right part of Figure~\ref{fig:weighted_BIS}.
Formally, we have.

\begin{definition} \label{def:gadget}
Let $\varphi$ be a Boolean formula in conjunctive normal form with variables $x_1,\ldots,x_n$ and clauses $c_1,\ldots,c_m$.
Moreover, let $G'_\varphi$ denote the bipartite graph from Definition \ref{def:auxgraph} with $2n+m$ vertices.
Further, let~$p$ be a prime, $\lweight, \rweight \in \zsp$ and $B$ be the bipartite graph with the distinguished vertices $u_L \in V_L(B)$ and $v_R \in V_\R(B)$ as provided by Lemma \ref{lem:bip-gadget}

For every $j \in [2n+m]$ denote by $B^j$ a copy of $B$ where every vertex $v \in V(B)$ is renamed $v^j$. The bipartite graph
$G_\varphi$ consists of the disjoint union of $G'_\varphi$ and $\bigcup_{j\in [2n+m]}\ B^j$ with the following identifications:
For all $i \in [n]$ identify $w_i$ with $u_L^i$ and $z_i$ with $v^{n+i}_R$.
For every $j \in [m]$ identify $y_j$ with $v_R^{2n+j}$.
\end{definition}

\smallskip
We observe that the graph $G_\varphi$ is bipartite. Moreover, the identification of the vertices is such that  the 
assignment of vertices to the partition is preserved, i.e., $v \in \lpart (G_\varphi)$ if and only if $v \in \lpart 
(G'_\varphi)$ or $v \in \lpart(B^j)$ for some $j \in [2n +m]$.
This is justified since vertices in $\lpart(G'_\varphi)$ are identified exclusively with vertices in $\lpart(B)$ and  
vertices in $\rpart(G'_\varphi)$ are identified exclusively with vertices in $\rpart(B)$ in the above construction.

The following partition will be useful in our proofs to follow.

\begin{definition}\label{def:wbis_gadget_partitioning}
Let $\varphi$ be a Boolean formula in conjunctive normal form with $n$ variables and $m$ clauses and let $G_\varphi$ be the associated bipartite gadget graph from Definition \ref{def:gadget}.
We recursively define a partition~$\{S_j\}_{j=0}^{2n+m}$ of $\calI(G_\varphi)$ by
\begin{align*}
S_1&:=\{\,I\in\calI(G_\varphi)\mid v_1,\vbar_1\notin I\,\} \\
S_j&:=\begin{cases}
\{\,I\in\calI(G_\varphi)\setminus \bigcup_{i=1}^{j-1} S_{i}\mid \Gamma_{G'_\varphi}(w_j) \cap I = \emptyset\,\} &\text{for } j \in \{2, \dots , n\}, \\
\{\,I\in\calI(G_\varphi)\setminus \bigcup_{i=1}^{j-1} S_{i}\mid \Gamma_{G'_\varphi}(z_{j-n}) \cap I = \emptyset\,\} &\text{for } j \in \{n+1, \dots , 2n\}, \\
\{\,I\in\calI(G_\varphi)\setminus\bigcup_{i=1}^{j-1} S_{i} \mid \Gamma_{G'_\varphi}(y_{j-2n}) \cap I = \emptyset\,\}& \text{for } j \in \{ 2n+1 , \dots , 2n+m\}.
\end{cases} \\
S_0&:=\calI(G_\varphi)\setminus {\textstyle \bigcup_{i=1}^{2n+m} S_{i}}.
\end{align*}
\end{definition}

For every $i \in [n]$, $\Gamma_{G'_\varphi}(w_i)=\{v_i,\vbar_i\}$, so for any independent set $I\in S_i$, both
$v_i,\vbar_i\notin I$. Similarly for every $i\in [n]$ and every $I\in S_{n+i}$, both $u_i,\ubar_i\notin I$. 
Additionally, for every $j \in [m]$, $S_{2n+j}$ does not contain independent sets of $G_\varphi$, which intersect with 
the neighbourhood $\Gamma_{G'_\varphi}(y_j)=\{u_i\mid x_i \text{ is a literal in } c_j\}\cup\{\ubar_i\mid \xbar_i \text{ 
is a literal in } c_j\}$. Consequently, $S_{0}$ contains any independent set~$I$ in $G_\varphi$, such that, for all 
$i\in[n]$, at least one of  $u_i$, $\ubar_i$ and at least one of $v_i$, $\vbar_i$ are in $I$ and, furthermore, for all 
$j\in[m]$, $\Gamma_{G'_\varphi}(y_j)\cap I\neq \emptyset $.

\smallskip
The following shows that the independent sets of every partition except~$S_0$, cancel out when counting modulo~$p$.

\begin{lemma}\label{lem:wbis_gadget_partitioning}
Let $\varphi$ be a Boolean formula in conjunctive normal form with $n$ variables and $m$ clauses and let $G_\varphi = \bipG$ be the associated bipartite gadget graph from Definition~\ref{def:gadget} as well as $\{S_j\}_{j=0}^{2n+m}$ the partition of $\calI(G_\varphi)$ as defined in Definition~\ref{def:wbis_gadget_partitioning}.
Then, for every $j\in[2n+m]$ 
\[
\sum_{I\in S_j}\lweight^{|\lpart\cap I|}\rweight^{|\rpart\cap I|}\equiv 0\pmod p.
\]
\end{lemma}
\begin{proof}
	We fix a $j\in[2n+m]$ and commence by defining the equivalence relation $\sim_j$ on $S_j$. For any two independent sets $I,I'\in S_j$ we have $I \sim_j I'$ if $I\setminus V(B^j)=I'\setminus V(B^j)$.
	That is, $I$ and $I'$ are equivalent if and only if they only differ in the vertices of $B^j$.
	We denote the $\sim_j$-equivalence class of $I$ by $\eqclass{I}_j$. Thus, $(\eqclass{I}_j)_{I\in S_j}$ is a partition of $S_j$.
	
	Let $I_1,\dots,I_{t_j}$ be representatives from each $\sim_j$-equivalence class. We obtain
	\begin{equation*}
		\sum_{I\in S_j} \lweight^{|\lpart\cap I|} \rweight^{|\rpart \cap I|}=\sum_{s=1}^{t_j} \sum_{I\in \eqclass{I_s}_j} \lweight^{|\lpart \cap I|} \rweight^{|\rpart \cap I|}.
	\end{equation*}
	Therefore, it suffices to establish $\sum_{I\in \eqclass{I_s}_j} \lweight^{|\lpart \cap I|} \rweight^{|\rpart \cap I|} \equiv 0 \pmod p$  for every $s \in [t_j]$.

	Let $I_s$ be one of the representatives $I_1,\dots,I_{t_j}$ with its associated equivalence class $\eqclass{I_s}_j$. We continue by studying the set $I_B = I_s \setminus V(B^j)$ of common vertices among the independent sets of $\eqclass{I_s}_j$. Therefore, every independent set $I \in \eqclass{I_s}_j$ contains the vertices in $I_B$. On the other hand, let $I'_B = \{ I \setminus I_B \mid I \in \eqclass{I_s}_j \}$ be the set of vertices in an independent set $I \in \eqclass{I_s}_j$, which are not in $I_B$. Since $B^j$ is a bipartite graph and the assignment of vertices to their relative component is conserved in the construction of $G_\varphi$ we obtain
	\begin{align}\label{eq:is-into-eq-classes}
		\sum_{I\in \eqclass{I_s}_j} \lweight^{|\lpart\cap I|} \rweight^{|\rpart \cap I|} & = \lweight^{|\lpart\cap I_B|} \rweight^{|\rpart\cap I_B|}
		\sum_{I \in I'_B } \lweight^{|\lpart(B^j)\cap I |} \rweight^{|\rpart(B^j)\cap I |} .
	\end{align}
	
	Let $x^j$ be the vertex of $B^j$ that is identified with one of the vertices of $G'_\varphi$ for the 
	construction of $G_\varphi$. Therefore, $x^j=u_\L^i$ if $j<n$, and $x^j=v_\R^i$ otherwise.
	By Definition~\ref{def:wbis_gadget_partitioning} we observe that for any $I\in S_j$ no neighbour of $x_j$ outside 
	$B^j$ is in $I$. 

	Hence, any vertex in $B^j$ is eligible for a construction of an independent set in $\eqclass{I}_j$. And vice versa, any independent set $I' \in \calI (B^j)$ yields an independent set in $\eqclass{I_s}_j$ by taking the union of $I'$ with $I_B$. We deduce that $I'_B = \calI (B^j)$. Therefore, the sum in the right hand side of \eqref{eq:is-into-eq-classes} is $\wISet{B^j}$. For this we recall that each $B^j$ was chosen utilizing Lemma~\ref{lem:bip-gadget}, whose property~$\mathit{1}$ yields $\wISet{B^j} \equiv 0 \pmod p$. We deduce the desired
	\begin{align*}
		\sum_{I\in \eqclass{I_s}_j} \lweight^{|\lpart\cap I|} \rweight^{|\rpart \cap I|}&= \lweight^{|\lpart\cap I_B|} \rweight^{|\rpart\cap I_B|} \wISet{B^j} \equiv 0 \pmod p ,
	\end{align*} 
	which proves the lemma.
\end{proof}

We have completed our setup and we are ready to prove the main result of this section.

{\renewcommand{\thetheorem}{\ref{thm:wbis-hardness}}
\begin{theorem}
Let $p$ be a prime and let $\lweight$, $\rweight \in\zp$. If $\lweight \equiv 0 \pmod p$ or $\rweight \equiv 0 \pmod p$ then \pbislr{\lweight}{\rweight} is computable in polynomial time. Otherwise, \pbislr{\lweight}{\rweight} is \numpp{}-complete.
\end{theorem}}
\begin{proof}
The first statement is a direct consequence of Proposition~\ref{prop:BIS_easy_cases}. Thus, let $\lweight, \rweight$ be in~$\zsp$.
We are going to show hardness for $\pbislr{\lweight}{\rweight}$ by establishing a Turing reduction from~\numpsat{}, which is known to be \numpp{}-complete by Simon's Theorem~\ref{thm:sat-hardness}.

Let $\varphi$ be a Boolean formula in conjunctive normal form with $n$ variables and $m$ clauses.
We show that the constructed bipartite graph $G_\varphi = \bipG$ from Definition~\ref{def:gadget} satisfies $\wISet{G_\varphi}\equiv K |\sat(\varphi)|\pmod p$ for some 
$K\not\equiv0\pmod p$. The exact value of $K$ depends on the values of the weights corresponding to the cases in the proof of Lemma~\ref{lem:bip-gadget}, but is not of interest for our argument.

Based on the partition~$\{S_j\}_{j=0}^{2n+m}$ given by Definition~\ref{def:wbis_gadget_partitioning}, we obtain
\begin{align}
\wISet{G_\varphi}&=\sum_{j=0}^{2n+m}\sum_{I\in S_j} \lweight ^{|\lpart \cap I|} \rweight^{|\rpart \cap I|}. \label{eq:is-partitioned}
\intertext{
By Lemma~\ref{lem:wbis_gadget_partitioning} every term of~\eqref{eq:is-partitioned} is equivalent to $0$ in $\zp$ \emph{except} the one regarding $S_0$.
This yields
}
\wISet{G_\varphi} &\equiv \sum_{I\in S_{0}} \lweight^{|\lpart\cap I|}\rweight^{|\rpart\cap I|}. \label{eq:sum-only-0}
\end{align}

As in the proof of Lemma~\ref{lem:wbis_gadget_partitioning} we are going to use an equivalence relation $\sim_0$ along with the associated equivalence classes $\eqclass{\cdot}_0$.
We define $U:=\{u_i,\ubar_i,v_i,\vbar_i\mid i\in[n]\}$ and the equivalence relation for two independent sets $I,I'\in S_{0}$ by $I\sim_0 I'$ if $I\cap U= I'\cap U$.
That is, $I$ and $I'$ have the same assignments of vertices in $U$.
Let $I_1,\dots,I_{t}$ be representatives for the $\sim_0$-equivalence classes.
We obtain
\begin{align}\label{eq:wbis_partitioning_S0}
\sum_{I\in S_{0}} \lweight^{|\lpart\cap I|}\rweight^{|\rpart\cap I|}& = \sum_{s=1}^{t}\sum_{I\in \eqclass{I_s}_0} \lweight^{|\lpart\cap I|} \rweight^{|\rpart\cap I|}.
\end{align}

Let $s \in [t]$ and $I \in \eqclass{I_s}_0$.
Since $I \in S_0$, at least one of $u_i,\ubar_i$ and at least one of $v_i,\vbar_i$ are in $I$.
We recall that for each $i\in[n]$ both $(u_i,v_i)$ and $(\ubar_i,\vbar_i)$ are edges in $G_\varphi$.
Therefore, either the pair $\{u_i,\vbar_i\} \subseteq I$ or the pair $\{\ubar_i,v_i \} \subseteq I$ and consequently, for each $i \in [n]$ neither $w_i \,(= u_\L^i)$ nor $z_i \,(= v_R^{n+i})$ can be in $I$.
Furthermore, for each $j\in[m]$ there exists at least one vertex in $\Gamma_{G'_\varphi}(y_j) \cap I$ by the definition of $S_0$.
Hence, for each $j \in [m]$ the vertex $y_j=v_R^{2n+j}$ cannot be in $I$.
We deduce that $I$ contains exactly $n$ vertices from $V_L(G'_\varphi)$ and exactly $n$ vertices from $V_R(G'_\varphi)$.

Each graph $B^j$ is a copy of the graph $B$ and the vertices $u_L^j$ for $j\leq n$ and $v_R^j$ for $j > n$, respectively, are cut vertices in $G_\varphi$.
There are $n$ copies of $B$ with $u_L$ identified with a vertex in $G'_\varphi$ and $n+m$ copies of $B$ with $v_R$ identified with a vertex in $G'_\varphi$.
Clearly, for arbitrary graphs $G_1$ and $G_2$ with disjoint vertex sets it holds $\wISet{G_1\cup G_2}=\wISet{G_1}\,\wISet{G_2}$.
This yields for every $s \in [t]$
\begin{align*}
\sum_{I\in \seqclass{I_s}_0} \lweight^{| \lpart \cap I|} \rweight^{|\rpart \cap I|}=
(\lweight \rweight)^n & \left(\sum_{I\in\calI(B-u_\L)} \lweight^{|\lpart(B-u_\L)\cap I|} \rweight^{|\rpart(B-u_\L)\cap I|}\right)^n \nonumber \\
&\left(\sum_{I\in\calI(B-v_\R)} \lweight^{|\lpart(B-v_\R)\cap I|} \rweight^{|\rpart(B-v_\R)\cap I|}\right)^{n+m}.
\end{align*}
Since $B$, $B-u_\L$ and $B - v_\R$ are bipartite graphs we obtain
\begin{align*}
\sum_{I\in \seqclass{I_s}_0}\lweight^{|\lpart\cap I|} \rweight^{|\rpart\cap I|}=(\lweight \rweight)^n &\left(\wISet{B-u_\L}\right)^n \left(\wISet{B-v_\R}\right)^{n+m}.
\end{align*}
We recall that $B$ was chosen due to Lemma~\ref{lem:bip-gadget}, whose Property~$\mathit{2}$ and Property~$\mathit{3}$ assure
\begin{align}\label{eq:equiv-class-0-is-K}
K := \sum_{I\in \seqclass{I_s}_0}\lweight^{|\lpart\cap I|} \rweight^{|\rpart\cap I|}& \;\not\equiv  0 \pmod p.
\end{align}

Combining \eqref{eq:equiv-class-0-is-K} and \eqref{eq:wbis_partitioning_S0} in conjunction with \eqref{eq:sum-only-0} we derive
\begin{align}\label{eq:sum-of-K}
\wISet{G_\varphi} \;\equiv\; tK \pmod p.
\end{align}

We will conclude the proof by constructing a bijection between the $\sim_0$-equivalence classes and the satisfying assignments of $\varphi$. In this manner we will obtain $t = |\sat(\varphi)\:\!|$.

For every equivalence class $\eqclass{I_s}_0$ with $s \in [t]$ we denote the set of common vertices in $\eqclass{I_s}_0$ by $U_s =\bigcap_{I\in\seqclass{I_s}_0}I$.
Due to the definition of $\sim_0$ for every $i \in [n]$ either the pair $\{u_i, \vbar_i\}$ or the pair $\{\ubar_i, v_i\}$ is shared by all elements of $\eqclass{I_s}_0$.
Hence, $U_s$ contains exactly $n$ such pairs of vertices.

Given an equivalence class~$\eqclass{I_s}_0$ utilizing $U_s$ we obtain an assignment~$a_s$ for~$\varphi$ by assigning for all $i \in [n]$
\[
x_i \mapsto 
\begin{cases}
\mathrm{true }, &\mathrm{if } \;\{u_i, \vbar_i\} \subseteq U_s; \\
\mathrm{false }, &\mathrm{if } \;\{\ubar_i,v_i\} \subseteq U_s.
\end{cases}
\]
We observe that each $\eqclass{I_s}_0$ yields a unique assignment $a_s$.
In order to show that it is a satisfying assignment it suffices to show that each clause of $\varphi$ is satisfied when we apply $a_s$.

Since $I_s \in S_0$, for each clause $c_j$ of $\varphi$ there exists at least one vertex $u \in \Gamma_{G'_\varphi}(y_j)$ with $u \in I_s$.
Due to the construction of $G_\varphi$ this vertex $u$ is either $u_i$, if $x_i$ appears non-negated in the clause $c_j$, or $\ubar_i$, if $x_i$ appears negated in the clause $c_j$.
Hence, $a_s$ satisfies $c_j$ at least once.

Vice versa, we now argue that every satisfying assignment can be obtained from an equivalence class $\eqclass{I_s}_0$ for some $s\in[t]$.
Let $a$ be a satisfying assignment for $\varphi$, this assignment gives rise to the set 
\[
U_a= \bigcup_{i \in [n]} \{u_i,\vbar_i \mid \textrm{if $x_i$ is set ``true'' by } a\}\cup\{\ubar_i,v_i\mid \textrm{if $x_i$ is set ``false'' by } a\}
\]
which is in $S_0$ and thus for $s$ such that $\eqclass{U_a}_0=\eqclass{I_s}_0$ it holds $a_s=a$.
		
We deduce that there are $t$ satisfying assignments of $\varphi$ and by \eqref{eq:sum-of-K}
\[
\wISet{G_\varphi}\equiv K|\sat(\varphi)|\pmod p,
\]
which establishes the theorem.
\end{proof}

%% file: polyt_computable_homs.tex
We identify the classes of graphs~$H$ for which \phcol{} can be solved in polynomial time. When counting graph 
homomorphisms modulo a prime~$p$, the automorphisms of order $p$ of a target graph~$H$ help us identify groups of 
homomorphisms that cancel out. More 
specifically assume that the target graph~$H$ has an automorphism~$\varrho$ of order~$p$. For any 
homomorphism~$\sigma$ from the input graph~$G$ to~$H$, $\sigma\circ\varrho$ is also a homomorphism from~$G$ to~$H$. 
This shows that the sets which contain the homomorphisms~$\sigma\circ\varrho^{(j)}$, for 
$j\in[p]$, have cardinality a multiple of~$p$, and thus, cancel out. This intuition is captured by the theorem of Faben 
and Jerrum~\cite[Theorem~3.4]{FJ13}. Before we formally state their theorem, we need the following definition.



\begin{definition}
 Let $H$ be a graph and $\varrho$ an automorphism of $H$. $H^\varrho$ is the subgraph of $H$ induced by the fixed 
points of $\varrho$.
\end{definition}




\begin{theorem}[Faben and Jerrum]\label{thm:involution-reduction}
Let $G$, $H$ be graphs, $p$ a prime and $\varrho$ an automorphism of $H$ of order~$p$. Then $|\Homs{G}{H}|\equiv|\Homs{G}{H^\varrho}|\pmod p$.
\end{theorem}

We can repeat the above reduction of $H$ recursively in the following way.

\begin{definition}
$H \Rightarrow_p H'$ if there is an automorphism~$\varrho$ of~$H$ of order~$p$ such that $H^\varrho=H'$.
We will also write $H \Rightarrow_p^* H'$ if either $H\isoto H'$ or, for some positive integer~$k$,  there are graphs 
$H_1, \dots, H_k$ such that $H \isoto H_1$, $H_1 \Rightarrow_p \cdots \Rightarrow_p H_k$, and
$H_k  \isoto H'$.
\end{definition}

Faben and Jerrum~\cite[Theorem~3.7]{FJ13} show for any choice of intermediate homomorphisms of order~$p$, the reduction 
$H \Rightarrow_p^* H'$ will end up in a unique graph up to isomorphism.

\begin{theorem}[Faben and Jerrum]\label{thm:involution-unique}
Given a graph~$H$ and a prime~$p$, there is (up to isomorphism) exactly one graph~$H^{*p}$ that has no automorphism of 
order~$p$ and $H \Rightarrow_p^* H^{*p}\!$.
\end{theorem}

The latter suggest the following definition.

\begin{definition}
 \label{defn:reduced-form}
 We call the unique (up to isomorphism) graph~$H^{*p}$, with $H \Rightarrow_p^* H^{*p}\!$, the \emph{order~$p$ reduced 
form} of~$H$.
\end{definition}

\begin{figure}[t]
\begin{center}
\includegraphics[]{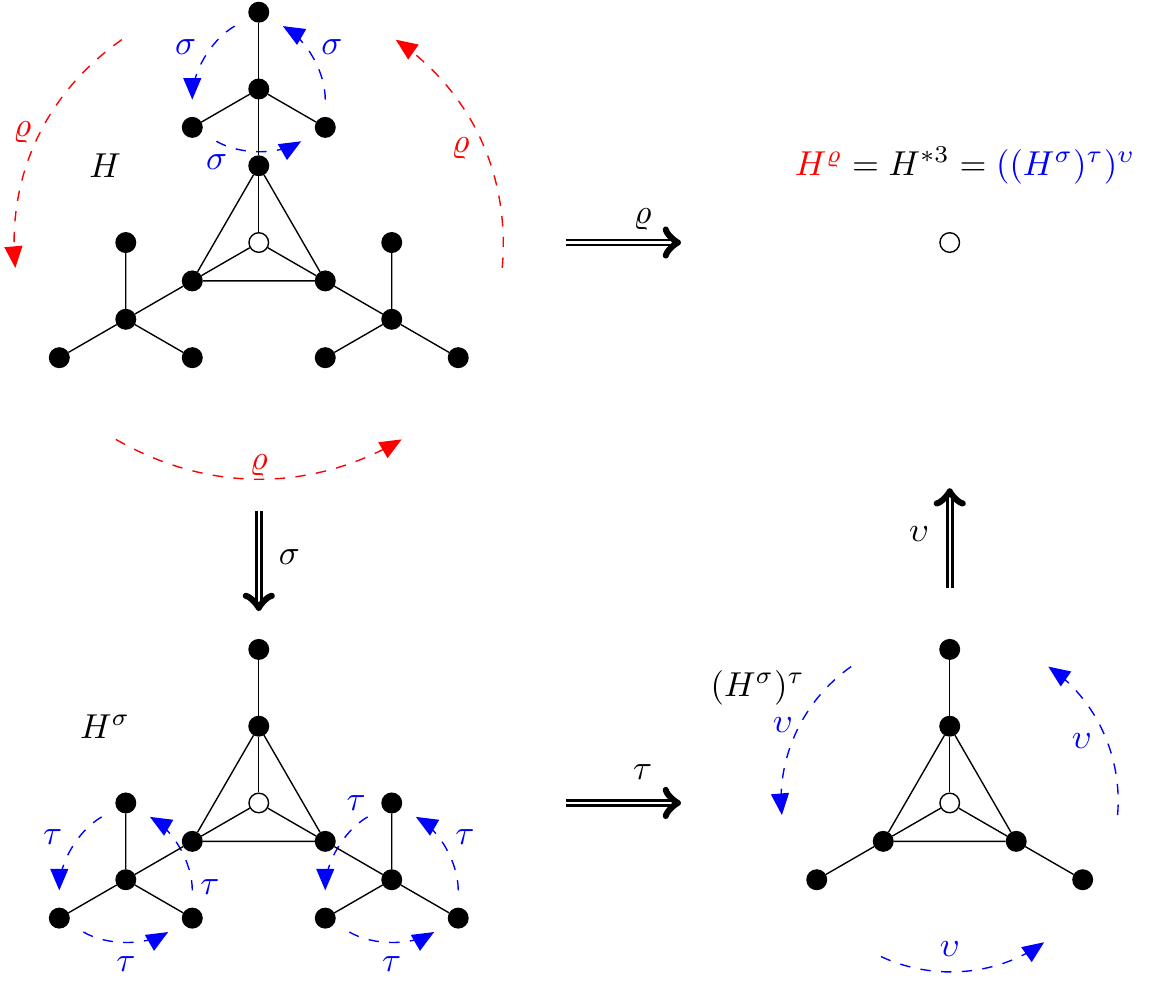}
\end{center}
\caption{An example of the order~3 reduced form~$H^{*3}$ of the graph $H$. Here we indicate two different ways of $H\Rightarrow_3^* H^{*3}$. The automorphism~$\varrho$ has order~3. It is indicated with red colour and $H^\varrho=H^{*3}$. $\sigma$, $\tau$ and $\upsilon$ each are automorphisms of order~3. These are indicated with blue colour and $((H^\sigma)^\tau)^\upsilon=H^{*3}$.}
\label{fig:reduced-form}
\end{figure}

Figure~\ref{fig:reduced-form} illustrates Theorem~\ref{thm:involution-unique} with an example of an order~3 reduced form of a graph.
Note that if $H$ has no loops the repeated application of the ``$\Rightarrow_p$'' operation does not introduce any 
loops.
 
\medskip
In order to compute the number of homomorphisms from~$G$ to~$H$ modulo~$p$, denoted by \phcol{}, it suffices to compute the number of homomorphisms from~$G$ to~$H^{*p}$ modulo~$p$.
To obtain the graphs for which \phcol{} is computed in polynomial time, we refer to the dichotomy theorem by Dyer and 
Greenhil~\cite[Theorem~1.1]{DG00}. 

\begin{theorem}[Dyer and Greenhil]\label{thm:dyer-greenhil}
Let $H$ be a graph that can contain loops.
If every component of $H$ is a complete bipartite graph with no loops or a complete graph with all loops present, 
then \nhcol{H} can be solved in polynomial time. Otherwise \nhcol{H} is $\shp$-complete.
\end{theorem}

We notice that a polynomial time algorithm for \nhcol{H}, gives a polynomial time algorithm for 
\phcol{} by simply applying the modulo~$p$ operation. In our setting, $H$ contains no loops, so we have the 
following characterisation for the polynomial time computable instances of \phcol{}.

\begin{corollary}\label{cor:polyt-graphs}
Let $H$ be a graph. If every component of $H^{*p}$ is a complete bipartite graph, then \phcol{} is computable in 
polynomial time.
\end{corollary}

%% file: part_labelled_graphs_pinning.tex
We prove that counting the number of homomorphisms from a partially labelled graph~$J$ to a fixed graph~$H$ 
modulo~$p$ reduces to the problem of counting homomorphisms from a graph~$G$ to~$H$ modulo~$p$.
This generalises the results of G\"obel, Goldberg and Richerby~\cite{GGR15}.
Many of the definitions and key lemmas we use in this sections are generalisation of the ones 
in~\cite[Section~3]{GGR15}, so our presentation follows the presentation of~\cite{GGR15} closely.

\bigskip
We study the following problem.

\probpar{\partlabphcol{}.}
{Graph~$H$ and prime~$p$.}
{Partially $H$-labelled graph $J=(G,\tau)$.}
{$|\Homst{J}{H}| \pmod p$.}

\medskip
According to \Lovasz{}, two graphs $H$ and~$H'$ are isomorphic if and only if for every graph~$G$ holds $|\Homs{G}{H}|  
= |\Homs{G}{H'}|$.
Faben an Jerrum ~\cite[4.5]{FJ13}, using a slightly different terminology, show that this results holds for partially 
labelled graphs~$J$ when the pinning function is restricted to maps exactly one vertex of $G(J)$ to a vertex of $H$, 
modulo all primes~$p$. G\"obel, Goldberg and Richerby~\cite[Lemma~3.6]{GGR15} show the following version of this result.
%

\begin{lemma}[G\"obel, Goldberg and Richerby]
\label{lem:Lovasz-2}
Let~$p$ be a prime and let $(H, \vbar)$ and $(H'\!, \vbar')$ be graphs that both have no automorphism of order $p$, each with $r$~distinguished vertices.  
Then $(H, \vbar) \isoto (H'\!,\vbar')$ if and only if, for all (not necessarily connected) graphs $(G,\ubar)$ with $r$~distinguished vertices,
\begin{equation*} 
|\Homs{(G,\ubar)}{(H,\vbar)}| \equiv |\Homs{(G,\ubar)}{(H'\!,\vbar')}| \pmod 2\,.
\end{equation*}
\end{lemma}

This version is more general than the result by Faben and Jerrum, in the sense that the pinning function can map any number of 
vertices, but it is only stated for modulo~2. A discussion about the subtle differences of the two results appears 
in~[Section~3.4]\cite{GGR15}. For our purposes, although the result of Faben and Jerrum suffices, we observe that the 
proof of Lemma~\ref{lem:Lovasz-2} holds modulo all primes~$p$.

\begin{lemma}
\label{lem:Lovasz}
Let~$p$ be a prime and let $(H, \vbar)$ and $(H'\!, \vbar')$ be graphs having no automorphism of order $p$, each with $r$~distinguished vertices.  
Then $(H, \vbar) \isoto (H'\!, \vbar')$ if and only if, for all (not necessarily connected) graphs $(G,\ubar)$ with $r$~distinguished vertices,
\begin{equation*} 
|\Homs{(G,\ubar)}{(H,\vbar)}| \equiv |\Homs{(G,\ubar)}{(H'\!,\vbar')}| \pmod p\,.
\end{equation*}
\end{lemma}
\noindent\emph{Explanation.}
In the proof of G\"obel et al.~\cite[Lemma~3.6]{GGR15} the following equation is shown.
\begin{align}    
|\InjHoms{(G,\ubar)}{(H,\vbar)}\,| &\equiv |\InjHoms{(G,\ubar)}{(H'\!,\vbar')}\,| \pmod 2\,. \nonumber
\intertext{
This is Equation~(2) from~\cite[Lemma~3.6]{GGR15}.
By reviewing the proof, one can observe that no modular equivalences are used, so the following equation holds.
}
|\InjHoms{(G,\ubar)}{(H,\vbar)}\,| &= |\InjHoms{(G,\ubar)}{(H'\!,\vbar')}\,|. \label{eq:injective}
\end{align}
Now we can show that \eqref{eq:injective}~holding for all graphs $(G, \ubar)$ with $r$ distinguished vertices implies that $(H,\vbar) \isoto (H'\!,\vbar')$.
To see this, consider $(G,\ubar) = (H,\vbar)$.
An injective homomorphism from a finite graph to itself is an automorphism and, since  $(H,\vbar)$ has no automorphism of order $p$, $\Aut(H,\vbar)$ has no element of order~$p$, so $|\Aut(H,\vbar)|\not\equiv0\pmod p$ by Cauchy's group theorem (Theorem~\ref{thm:Cauchy}).
By~\eqref{eq:injective}, the number of injective homomorphisms from $(H,\vbar)$ to $(H'\!,\vbar')$ is not equivalent to $0\pmod p$, which means that there is at least one such homomorphism.
Similarly, taking $(G,\ubar) = (H'\!,\vbar')$ shows that there is an injective homomorphism from $(H'\!,\vbar')$ to 
$(H,\vbar)$ and therefore, the two graphs are isomorphic.\qed

\medskip
A complete, self-contained proof of Lemma~\ref{lem:Lovasz} can also be found in~\cite{GobThesis}.

\medskip
As in~\cite{GGR15} we introduce orbit vectors, but generalised to an arbitrary prime~$p$.


\begin{definition}
\label{defn:lambdas}
Let $H$ be a graph with no automorphism of order~$p$ and $r\in\Z_{>0}$.
An enumeration $\vbar_1,\dots,\vbar_\mu$ of elements of $(V(H))^r$ such that, for every $\vbar\in(V(H))^r\!$, there is 
exactly one $i\in[\mu]$ such that $(H,\vbar)\isoto (H,\vbar_i)$ is referred to as an \emph{enumeration of~$(V(H))^r$ up 
to isomorphism}.
\end{definition}

The number~$\mu$ of tuples in the enumeration depends on the structure of~$H$ and not only on $|V(H)|$.

\begin{definition}\label{defn:vector}
Let $H$ be a graph with no automorphism of order~$p$, $r\in\Z_{>0}$ and let $\vbar_1, \dots, \vbar_\mu$ be an enumeration of $(V(H))^r$ up to isomorphism.
Further, let $(G,\ubar)$~be a graph with $r$~distinguished vertices.
We define the \emph{orbit vector} $\vecv_H(G, \ubar)\in(\Z_p)^\mu$ where, for each $i\in[\mu]$, the $i$-th component of $\vecv_H(G, \ubar)$ is given by
\begin{equation*}
\big(\vecv_H(G, \ubar)\big)_i \equiv |\Homs{(G,\ubar)}{(H,\vbar_i)}| \pmod p\,.
\end{equation*}
We say that $(G,\ubar)$ \emph{implements} this vector.
\end{definition}

For a group~$\calG$ acting on a set~$X$, the \emph{orbit} of an element $x\in X$ is defined to be the set 
$\Orb_{\calG}(x) = \{\pi(x)\mid \pi\in\calG\}$. 
For a graph~$H$, we will abuse notation, writing $\Orb_H(\cdot)$ instead of $\Orb_{\Aut(H)}(\cdot)$.
Thus, for $r\in\Z_{>0}$ and an enumeration $\vbar_1, \dots, \vbar_\mu$ of $(V(H))^r$ up to isomorphism, $|\{\,\vbar 
\in (V(H))^r \mid (H,\vbar)\isoto(H,\vbar_i) \,\}|=|\Orb_H(\vbar_i)|$ for every $i\in[\mu]$.

\medskip
Defining the vectors using the enumeration up to isomorphism hides the size of the orbit of a tuple $\vbar_i\in(V(H))^r$, as each orbit gets contracted to a single entry.
This information is not needed when counting modulo~2, as we can prove that for every tuple~$\vbar_i$, $|\Orb_H(\vbar_i)|$ is odd.
In contrast, this information is needed when counting modulo an odd prime.
We can recover this information at any point, since $H$ is fixed, as we are going to do later on. As it is more 
convenient to proof the technical lemmas using the contracted vectors of Definition~\ref{defn:vector} we will make this 
recovery at a later, more convenient point.

\medskip
Due to Lemma~\ref{lem:Lovasz}, for every graph $H$ and for all $\vbar \in (V(H))^r$ and $i\in[\mu]$ such that $(H,\vbar) \isoto (H,\vbar_i)$, we have that $\big(\vecv_H(G, \ubar)\big)_i\equiv|\Homs{(G,\ubar)}{(H,\vbar)}| \pmod p$.

\medskip
We denote by $\oplus^p$ and~$\otimes^p$ componentwise addition and multiplication modulo~$p$, of vectors in 
$(\zp)^\mu$, respectively.

\begin{lemma}
\label{lem:otimesvec}
Let $(G_1,\ubar)$, $(G_2,\ubar)$ be graphs, where $\ubar = u_1\dots u_r$ with $r\in\Z_{>0}$, such that $V(G_1)\cap V(G_2) = \{u_1,\dots, u_r\}$. 
Further, let $H$ be a graph with no automorphism of order~$p$ with an enumeration of $(V(H))^r$ up to isomorphism.
Then
\begin{equation*}
\vecv_H(G_1\cup G_2, \ubar) = \vecv_H(G_1, \ubar) \otimes^p \vecv_H(G_1, \ubar)\,.
\end{equation*}
\end{lemma}
\begin{proof}
A function $\sigma\colon V(G_1)\cup V(G_2)\to V(H)$ is a homomorphism from $(G_1\cup G_2, \ubar)$ to~$(H,\vbar)$ if and only if, for each $i\in\{1,2\}$, the restriction of~$\sigma$ to $V(G_i)$ is a homomorphism from $(G_i,\ubar)$ to~$(H,\vbar)$.
\end{proof}

Componentwise multiplication of $\vecv_H(G_1,\ubar_1)$ and $\vecv_H(G_2,\ubar_2)$ for 
two given graphs $(G_1,\ubar_1)$ and $(G_2, \ubar_2)$ can be expressed as an orbit vector of a single graph.
This is more complex for componentwise addition $\vecv_H(G_1,\ubar_1) \oplus^p \vecv_H(G_2,\ubar_2)$. For our purposes it is sufficient that a set of graphs whose vectors sum to a desired vector exists, componentwise.

For graphs with distinguished vertices $(G_1,\ubar_1), \dots, (G_t, \ubar_t)$, we define
\begin{equation*}
\vecv_H\big((G_1,\ubar_1) + \dots + (G_t,\ubar_t)\big) = \vecv_H(G_1,\ubar_1) \oplus^p \dots \oplus^p 
\vecv_H(G_t,\ubar_t)
\end{equation*}
and say that a vector $\vecv\in(\zp)^\mu$ is \emph{$H$-implementable}, if it can be expressed as such a sum.

\medskip
The modulo~2 version of the following lemma appears in~\cite[Lemma~4.16]{FJ13} and is used for all pinning techniques 
so far. We reprove the lemma for the vectors in $(\mathbb{Z}_p)^\mu$ when~$p$ is an arbitrary prime.

\begin{lemma}\label{lem:implementable-vectors}
Let $\mu\in\Z_{>0}$ and $S\subseteq(\mathbb{Z}_p)^\mu$ be closed under $\oplus^p$ and $\otimes^p$.
If $1^\mu\in S$ and, for every distinct $i,j\in[\mu]$, there is a tuple $s=s_1\dots s_\mu\in S$ with $s_i\neq s_j$, then $S=(\zp)^\mu$.
\end{lemma}
\begin{proof}
It suffices to show that all of the basis vectors of the standard basis
\!\!\footnote{The standard basis is the set $\{100\dots00,010\dots00,\dots,000\dots01\}$} in $(\mathbb{Z}_p)^\mu$ are in $S$.
Since $S$ is closed under $\oplus^p$ and $\otimes^p$ it follows that all of $(\mathbb{Z}_p)^\mu$ is in $S$. 

We show that all the basis vectors are in $S$ by induction on $\mu$.
If $\mu=1$ the lemma clearly holds as the all-ones vector is the only vector in the standard basis.
Assume that the induction hypothesis holds for $\mu-1$ and $\mu>1$. 
Then we can construct vectors that agree with the standard basis in the first $\mu-1$ places without being able to control what happens in the $\mu$-th place.
From the latter and the statement of the lemma, that $1^\mu\in S$, we obtain the following vectors
$$
\begin{array}{cccccccc}
 \vecv_0&=&1&1&1&\dots&1&1\\
 \vecv_1&=&1&0&0&\dots&0&x_1\\
 \vecv_2&=&0&1&0&\dots&0&x_2\\
 \vdots& \vdots & \vdots& & \ddots & \ddots & \vdots & \vdots\\
 \vecv_{\mu-1}&=&0&0&0&\dots&0&x_{\mu-1}\\
 \vecv_\mu&=&0&0&0&\dots&1&x_\mu
\end{array}
$$
where the $x_i$ can take any value in $\mathbb{Z}_p$.

Let $r$ be an integer and let $\vecv\in(\zp)^\mu$.
We use the notation $\vecv^r=\vecv\otimes^p\dots\otimes^p\vecv$ for the $r$-fold componentwise product and let 
$r\vecv=\vecv\oplus^p\dots\oplus^p\vecv$ denote the $r$-fold componentwise sum of $\vecv$.
Consider the values of each $x_i$.
If $x_i\neq0$, by Theorem~\ref{thm:Fermat} we have $x_i^{p-1}\equiv 1\pmod p$. 
Hence $\vecv_i^{p-1}=00\dots010\dots01$.
So from now on we can assume that for each $i\in[\mu]$, $x_i\in\{0,1\}$. 
We have the following three cases.

Case~1. For all $i\in[\mu]$, $x_i=0$. Then the vector $\vecv=\vecv_0\oplus^p\Bigop\limits_{i\in[\mu]} (p-1)\vecv_i=0\dots01$ is the remaining vector that completes the standard basis.

Case~2. There are at least two $i,j$ such that $x_i,x_j=1$.
Then $\vecv=\vecv_i\otimes^p \vecv_j=0\dots01$.
To obtain the remaining vectors of the standard basis, for each $i\in[\mu]$ with $x_i\neq0$, we take the vector $\vecv_i\oplus^p (p-1)\vecv$.

Case~3. There is exactly one $i\in[\mu]$ with $x_i=1$.
From the statement of the lemma there is a vector $\vecu\in S$, with $(\vecu)_i=a$, $(\vecu)_\mu=b$, where $a\neq b$.
First assume that $a>b$.
Let $\vecu_i=\vecu\otimes^p \vecv_i=0\dots0a0\dots0b$ and let $\vecv_a=(p-a)\vecv_i=0\dots0(p-a)0\dots0(p-a)$.
Then $\vecu_i\oplus^p\vecv_a=0\dots0(p-a+b)$.
Since $a>b$, $(p-a+b)$ is not a multiple of $p$ so, by Theorem~\ref{thm:Fermat}, $(p-a+b)^{p-1}\equiv1\pmod p$.
Thus, $\vecv=(\vecu_i\oplus^p\vecv_a)^{p-1}=0\dots 01$ and $\vecv_i'=(p-1)\vecv\oplus^p \vecv_i=0\dots 010\dots 0$ complete the standard basis.

Now assume that $a<b$.
Let $\vecv_b=(p-b)\vecv_i=0\dots0(p-b)0\dots0(p-b)$ and therefore $\vecu_i\oplus^p\vecv_b=0\dots0(p+a-b)0\dots0$.
Since $a<b$, $(p+a-b)$ is not a multiple of $p$ so, by Theorem~\ref{thm:Fermat}, $(p+a-b)^{p-1}\equiv1\pmod p$.
Thus $\vecv_i''=(\vecu_i\oplus^p\vecv_b)^{p-1}=0\dots 010\dots0$ and $\vecw=(p-1)\vecw\oplus^p \vecv_i=0\dots 01$ complete the standard basis.
\end{proof}


\begin{corollary}
\label{cor:implementable}
Let $H$ be a graph with no automorphism of order~$p$ with an enumeration $\vbar_1,\ldots,\vbar_\mu$ of $(V(H))^r$ up to isomorphism.
Then every $\vecv\in (\zp)^\mu$ is $H$-implementable.
\end{corollary}
\begin{proof}
Let $S$ be the set of $H$-implementable vectors.  $S$~is clearly closed under~$\oplus^p$, and is closed under~$\otimes^p$ by Lemma~\ref{lem:otimesvec}.
Let $G$~be the graph on vertices $\{u_1, \dots, u_r\}$, with no edges.
$1^\mu$~is implemented by $(G,u_1, \dots, u_r)$, which has exactly one homomorphism to every $(H,\vbar_i)$.  
Finally, for every pair $i,j\in [\mu]$, such that $(H,\vbar_i)$ and $(H,\vbar_j)$ are not isomorphic, by Lemma~\ref{lem:Lovasz}, there is a graph $(G,\ubar)$ such that
\begin{equation*}
|\Homs{(G,\ubar)}{(H,\vbar_i)}|\not\equiv |\Homs{(G,\ubar)}{(H,\vbar_j)}| \pmod p\,.
\end{equation*}
$(G,\ubar)$ implements a vector~$\vecv$ whose $i$th and $j$th components are different and the corollary follows from Lemma~\ref{lem:implementable-vectors}.
\end{proof}


At this point we have shown that all orbit vectors in $(\zp)^\mu$ are $H$-implementable. We can now define the tuple 
vectors that have an entry for each $r$-tuple. The tuple vectors include the sizes of the orbits 
$\Orb_H(\vbar)$, for all $v\in V(H)^r$, as this information is vital for the proof of our main theorem.

\begin{definition}\label{defn:tuple-vector}
Let $H$ be a graph with no automorphism of order~$p$, $r\in\Z_{>0}$ and let $\wbar_1, \dots, \wbar_\nu$ be an enumeration of $(V(H))^r$, i.e., $\nu=|V(H)|^r$.
Let $(G,\ubar)$~be a graph with $r$~distinguished vertices.
We define the \emph{tuple vector} $\vecw_H(G, \ubar)\in(\Z_p)^\nu$ where, for each $j\in[\nu]$, the $j$-th component of $\vecw_H(G, \ubar)$ is given by
\begin{equation*}
\big(\vecw_H(G, \ubar)\big)_j \equiv |\Homs{(G,\ubar)}{(H,\wbar_j)}| \pmod p\,.
\end{equation*}
We say that $(G,\ubar)$ \emph{implements} this vector.
\end{definition}

\begin{definition}
Let $H$ be a graph with no automorphism of order~$p$, $r\in\Z_{>0}$ and let $\wbar_1, \dots, \wbar_\nu$ be an enumeration of $(V(H))^r$, i.e., $\nu=|V(H)|^r$.
Denote by $F(H,r)\subseteq (\zp)^\nu$ the set of vectors~$\vecw$, such that, for all $i,j\in[\nu]$ with $(H,\wbar_i)\isoto (H,\wbar_j)$, we have $(\vecw)_i=(\vecw)_j$.
\end{definition}

The following lemma shows which tuple vectors are $H$-implementable. The proof uses the $H$-implementable orbit vectors 
and retracts the information that gets lost by using the enumeration up to isomorphism of the $r$-tuples.

\begin{lemma}\label{cor:implementable-full}
Let $H$ be a graph with no automorphism of order~$p$, $r\in\Z_{>0}$ and $\wbar_1, \dots, \wbar_\nu$ an enumeration of $(V(H))^r$, i.e., $\nu=|V(H)|^r$.
Then every $\vecw\in F(H,r)$ is $H$-implementable.
\end{lemma}
\begin{proof}
Let $\vbar_1,\ldots,\vbar_\mu$ be an enumeration  up to isomorphism of $(V(H))^r$. 
We denote by $f:[\mu]\to[\nu]$ the associated function with $\vbar_i=\wbar_{f(i)}$ for all $i\in[\mu]$, i.e., $f$ tells us which coordinates of the tuple vector are representatives for the equivalence classes giving the coordinates of the orbit vector.
Now, given $\vecw\in F(H,r)$, we compute the corresponding vector $\vecv \in (\Z_p)^\mu$ by letting 
$(\vecv)_i=(\vecw)_{f(i)}$ for all $i\in[\mu]$.
The vector $\vecv$ is $H$-implementable by Corollary~\ref{cor:implementable}.
Now, if $(G,\ubar)$ is a graph with $r$ distinguished vertices such that $(\vecv)_i \equiv |\Homs{(G,\ubar)}{(H,\vbar_i)}| \pmod p$ for all $i\in[\mu]$, then we also have $(\vecw)_j\equiv |\Homs{(G,\ubar)}{(H,\wbar_j)}| \pmod p$ for all $j\in[\nu]$.
\end{proof}

Before we prove the main theorem of this section, we need the following lemma.

\begin{lemma}
\label{lem:sumvec}
Let $H$ be a graph with no automorphism of order~$p$, $r\in\Z_{>0}$ and $\wbar_1, \dots, \wbar_\nu$ an enumeration of $(V(H))^r$, i.e., $\nu=|V(H)|^r$.
Then for every graph $(G,\ubar)$ with $r$ distinguished vertices
\begin{equation*}
|\Homs{G}{H}|\equiv \sum_{j\in[\nu]} (\vecw_H(G,\ubar))_j \pmod p\,.
\end{equation*}
\end{lemma}
\begin{proof}
We have,
\begin{align*}
\sum_{j \in [\nu]} (\vecw_H (G, \ubar))_j &\equiv \sum_{j \in [\nu]} 
|\Homs{(G,\ubar)}{(H,\wbar_j)}| \pmod p\ \\
&= |\Homs{G}{H}| \pmod p.
\end{align*}
The equivalence holds by the definition of $\vecw_H(G,\ubar)$.
The equality holds because every homomorphism from $G$ to~$H$ must map~$\ubar$ to some $r$-tuple~$\wbar$.
Since $[\nu]$ contains all $r$-tuples we obtain all homomorphisms from $G$ to $H$.
\end{proof}

{\renewcommand{\thetheorem}{\ref{thm:partlabcol}}
\begin{theorem}
Let $p$ be a prime and let $H$ be a graph.
Then \partlabphcol{} reduces to \phcol{} via polynomial time Turing reduction.
\end{theorem}}
\begin{proof}
Let $J=(G, \tau)$ be an instance of \partlabphcol{}.
Let $\ubar = u_1 \dots u_r$ be an enumeration of $\dom(\tau)$ and let $\wbar= w_1 \dots w_r = \tau(u_1) \dots \tau(u_r)$.
Moving from the world of partially $H$-labelled graphs to the equivalent view on graphs with distinguished vertices, we wish to compute $|\Homs{(G,\ubar)}{(H,\wbar)}|$ modulo~$p$.
Let $\wbar_1, \dots , \wbar_{\nu}$ be an enumeration of $(V(H))^r$ and let $\vecw \in \{0,1\}^\nu$ be the vector with $(\vecw)_j=1$ if $(H,\wbar_j)\isoto (H,\wbar)$, and $0$ for all other $j\in[\nu]$; $\vecw$ has exactly $|\Orb_H(\wbar)|$ 1-entries.
Since $\vecw\in F(H,r)$, by Lemma~\ref{cor:implementable-full} $\vecw$ is \linebreak[4] {$H$-implemented} by some sequence $(\Theta_1, \ubar_1), \dots , (\Theta_t, \ubar_t)$ of graphs with $r$-tuples of distinguished vertices.

For each $s \in [t]$, let $(G_s, \ubar)$ be the graph that results from taking the disjoint union of a copy of $G$ and $\Theta_s$ and identifying the $i$-th element of $\ubar$ with the $i$-th element of $\ubar_s$ for each $i \in [r]$.
Then Lemma~\ref{lem:otimesvec} yields
\begin{align}
\vecw_H (G_s, \ubar) &= \vecw_H (G, \ubar) \otimes^p \vecw_H (\Theta_s, \ubar_s). \nonumber 
\intertext{
With this we obtain
}
\vecw_H (G, \ubar) \otimes^p \vecw &= \vecw_H (G, \ubar) \otimes^p \vecw_H \left( (\Theta_1, \ubar_1) + \dots + (\Theta_t, \ubar_t) \right) \nonumber \\
&= \vecw_H (G, \ubar) \otimes^p \left( \vecw_H(\Theta_1, \ubar_1) \oplus^p \dots \oplus^p \vecw_H(\Theta_t, \ubar_t) \right) \nonumber \\
&= \Crossp_{s\in[t]} \left( \vecw_H (G, \ubar) \otimes^p \vecw_H (\Theta_s, \ubar_s) \right) \nonumber \\
&= 
\Crossp_{s\in[t]} \vecw_H (G_s, \ubar) \,. \nonumber
\intertext{
By summing the components of the vector $\vecw_H (G, \ubar) \otimes^p \vecw$, since $\vecw$ contains a 1-entry for each $\wbar_k\in\Orb_H(\wbar)$ and a 0-entry everywhere else, we have,
}
\sum_{j\in[\nu]}\left(\vecw_H (G, \ubar) \otimes^p \vecw\right)_j
&=|\Orb_H(\wbar)|\cdot|\Homs{(G,\ubar)}{(H,\wbar)}|.
\label{eq:sumvec-left}
\intertext{
Summing the components of the vector $\Crossp_{\!\!\!\! s\in[t]} \vecw_H (G_s, \ubar)$, we have
}
\sum_{j\in[\nu]}\left(\Crossp_{s\in[t]} \vecw_H (G_s, 
\ubar)\right)_j
&=\sum_{s\in[t]}\sum_{j\in[\nu]}\left(\vecw_H(G_s,\ubar)\right)_j
\label{eq:sumvec-right}
\end{align}

By applying Lemma~\ref{lem:sumvec}, we have that the values of~\eqref{eq:sumvec-right} are modulo $p$ congruent to $\sum_{s \in [t]} |\Homs{G_s}{H}|$.
Thus, from the equality of \eqref{eq:sumvec-left} and \eqref{eq:sumvec-right} we have
\[
|\Orb_H(\wbar)|\cdot|\Homs{(G, \ubar)}{(H, \wbar)}| = \, \sum_{s \in [t]} |\Homs{G_s}{H}|.
\]
The right side can be computed by making $t$ calls to an oracle for \phcol{}.
Since $H$ is fixed and $r$ is finite, we can trivially compute $|\Orb_H(\wbar)|$ and thus being able to recover $|\Homs{(G, \ubar)}{(H, \wbar)}|$ concludes the proof. 
\end{proof}

%% file: treegadgets.tex
We provide the classes of trees~$H$, for which~\phcol{} is \numpp-hard, utilising the previous theorem (Theorem~\ref{thm:partlabcol}) on \partlabphcol{}.
Due to Section~\ref{sec:poly} we focus on graphs that have no automorphism of order~$p$. 
Employing Corollary~\ref{cor:polyt-graphs} we can see that stars are graphs for which \phcol{} is tractable. A tree that is 
not a star contains a path of length at least $3$.
This path is the structure that will eventually give us hardness. We formally define.

\begin{definition}\label{def:ab-paths}
	Let $H$ be a graph, $p$ be a prime and $a,b\in\zp \setminus \{1\}$. Assume $H$ contains a path $P=x_0\dots x_k$ for $k>0$, such that the following hold
\begin{enumerate}
	\item $P$ is the unique path between $x_0$ and $x_k$ in $H$.
	\item $\deg_H(x_0)\equiv a\pmod p$ and $\deg_H(x_k)\equiv b \pmod p$.
	\item For all $0<i<k$, $\deg_H(x_i)\equiv 1 \pmod p$.
\end{enumerate}
Then, we will call $P$ an \emph{$(a,b,p)$-path in $H$} and denote it $\Path[H][a,b,p]$.
\end{definition}

We proceed by showing that every non-star tree~$H$ without automorphisms of order~$p$ contains such a path.

\begin{lemma}\label{lem:tree_star_or_hard}
	Let $H$ be a tree that has no automorphism of order~$p$. Then, either $H$ is a star or there are $a,b\in\Z_p\setminus\{1\}$ such that $H$ contains an $(a,b,p)$-path.
\end{lemma}
\begin{proof}
	We assume that $H$ is not a star and let $P=x_{-1}x_0\dots x_\ell$ be a maximal path of $H$ with length~$\ell+1$.
	We are going to prove that $P$ contains an $(a,b,p)$-path.  
	
	Since $H$ is not a star, $P$ contains at least four vertices yielding $\ell>1$.
	In order to prove that any vertex in $\Gamma_H(x_0)-x_1$ must be a leaf we assume the contrary.
	Let $v\in\Gamma_H(x_0)-x_1$ be not a leaf and $v'\neq x_0$ be a neighbour of $v$.
	Then, $v'vx_0\dots x_\ell$ is a path of length $\ell+2$ contradicting the maximality of $P$.
	The very same argument yields that any vertex in $\Gamma_H(x_{\ell-1})-x_{\ell-2}$ must be a leaf as well.
	
	We assume towards a contradiction that $|\Gamma_H(x_0)| > p$.
	Let $Y=\{y_1,\dots,y_p\} \subseteq \Gamma_H(x_0) - x_1$ be a set of neighbours of $x_0$, which are not equal to $x_1$. Let $\tau$ be a mapping from $H$ to itself defined as follows: for every vertex $y_i \in Y$, let $\tau(y_i) = y_{i+1}$ with the indices taken modulo $p$; for any other vertex $v \in V(H) \setminus Y$, let $\tau(v) = v$.
	As we have observed above, for all $i\in[p]$, $y_i$ is a leaf only adjacent to $x_0$.
	Therefore, $\tau$ is an automorphism of $H$ of order $p$, which is a contradiction.
	
	Hence, $x_0$ has at least two and at most $p$ neighbours, which yields $\deg_H(x_0)\not\equiv 1\pmod p$.
	Similarly, we obtain $\deg_H (x_{\ell-1}) \not\equiv 1\pmod p$.
	Consequently, there exists the minimum 
	\[
	\min \{k \in [\ell-1] \mid \deg(x_k)\not\equiv 1 \pmod p \},
	\]
	which yields the subpath $P'=x_0\dots x_k$ of $P$.
	Let $a=\deg_H(x_0)$ and $b=\deg_H (x_k)$.
	Since $H$ contains no cycles, $P'$ is the unique path in $H$ connecting $x_0$ and $x_k$.
	As we have obtained, $a, b \not\equiv 1 \pmod p$. Finally, due to the choice of $k$ we deduce $\deg_H (x_i) \equiv 1 \pmod p$ for all internal vertices $x_i$ of $P'$ with $i \in [k-1]$.
	We conclude that $P'$ is an $(a,b,p)$-path in $H$. 
\end{proof}

Before we study the consequences of an $(a,b,p)$-path $\Path[H]$ in $H$, in the next lemma we observe that the number of homomorphisms from a $(k+1)$-path to $H$ with two distinguished vertices and the number of possibilities $H$ offers to get from one of the distinguished vertices to the other, are equal.

\begin{lemma}\label{lem:walks}
	Let $p$ be a prime, $H$ a graph and let $x,y\in V(H)$. If $P$ is the path $z_0z_1\dots z_k$, then $|\,\Homs{(P,z_0,z_k)}{(H,x,y)}|$ 
	is equal to the number of $k$-walks in $H$ from~$x$ to~$y$.
\end{lemma}
\begin{proof}
	Let $W(x,y,k)$ denote the number of $k$-walks in $H$ between the vertices $x$ and $y$. We prove the lemma by induction on $k$.
	
	In the base case with $k=1$ the path $P$ consists only of the edge $(z_0,z_1)$. If $x$ is adjacent to $y$ in $H$, then there is only one homomorphism $\sigma: (P,z_0,z_k) \to (H,x,y)$ implying $W(x,y,1)=1=|\,\Homs{(P,z_0,z_1)}{(H,x,y)}|$. Otherwise, $x$ and $y$ are not adjacent. Hence, there cannot exist a homomorphism $\sigma: (P,z_0,z_1) \to (H,x,y)$ implying $W(x,y,1)=0=|\,\Homs{(P,z_0,z_k)}{(H,x,y)}|$. 

	Regarding the induction step, we assume $W(x,y,i)=|\,\Homs{(P,z_0,z_i)}{(H,x,y)}|$ holds for all paths $P$ of size $i < k$ and are going to show $W(x,y,k)=|\,\Homs{(P,z_0,z_k)}{(H,x,y)}|$. Let $W$ be a $k$-walk in $H$ from $x$ to $y$. In order to reach $x$ the walk $W$ must traverse a neighbour $u$ of $x$. Deleting the edge $(u,x)$ from $W$ yields a walk of length $k-1$ from $u$ to $y$. However, if for a neighbour $u'$ of $x$ there exists no $k$-walk from $x$ to $y$ traversing $u'$, then there is no $(k-1)$-walk from $u'$ to $y$. This yields 
	\begin{equation}\label{eq:sum_of_subwalks}
		W(x,y,k)=\sum_{u\in\Gamma_H(x)} W(u,y,k-1).
	\end{equation}
	
	Let $P'=z_1\dots z_{k}$ be the path obtained from $P$ by deleting the edge $(z_0, z_1)$. Since $z_1$ is adjacent to $z_0$ and every homomorphism $\sigma: (P,z_0,z_k) \to (H,x,y)$ maps $z_0$ to $x$, $z_1$ must be mapped to a neighbour of $x$. Hence, for every neighbour $u$ of $x$ a homomorphism $\sigma': (P',z_1,z_{k}) \to (H,u,y)$ yields a homomorphism from $(P,z_0,z_k)$ to $(H,x,y)$ and vice versa. We deduce
	\begin{equation}\label{eq:homs_sum_of_subwalks}
		|\,\Homs{(P,z_0,z_{k})}{(H,x,y)}|=\sum_{u\in\Gamma_H(x)}|\,\Homs{(P',z_1,z_{k})}{(H,u,y)}|.
	\end{equation}
	
	Finally, by \eqref{eq:homs_sum_of_subwalks}, the induction hypothesis and \eqref{eq:sum_of_subwalks} we obtain the desired \[
		|\,\Homs{(P,z_0, z_k)}{(H,x,y)}|=\sum_{u\in\Gamma_H(x)} W(u,y,k-1)=W(x,y,k).
	\qedhere
	\]
\end{proof}

\begin{corollary}\label{cor:distance}
	Let $G,H$ be graphs and let $u,v\in V(G)$. Then for every homomorphism $\sigma: G \to H$ holds $\dist_H(\sigma(u),\sigma(v))\leq\dist_G(u,v)$.
\end{corollary}
\begin{proof}
	We assume towards a contradiction that there exists a homomorphism $\sigma$ from $G$ to $H$ with $\dist_G(u,v)<\dist_H(\sigma(u),\sigma(v))$. Let $k=\dist_G(u,v)$. Since the distance between $\sigma(u)$ and $\sigma(y)$ in $H$ is larger than $k$, there exists no $k$-walk in $H$ between $\sigma(u)$ and $\sigma(y)$. Therefore, by Lemma~\ref{lem:walks} $\sigma$ cannot exist.
\end{proof}


	In order to show that \phcol{} is \numpp-hard we are going to establish a reduction from 
\pbislr{\lweight}{\rweight} to \partlabphcol{}. That is, 
	given a graph $G$ input for \pbislr{\lweight}{\rweight}, we construct a partially labelled 
graph $J$, input for \partlabphcol{}, such that $\wISet{G}\equiv |\,\Homs{J}{H}| \pmod p$.
	The construction of $J$ as stated uses any path in $H$. When we define the actual reduction though, we will 
require that this path is an $(a,b,p)$-path in~$H$.
		
	Let $G=\bipG$ be the bipartite input graph of \pbislr{\lweight}{\rweight}, let $p$ be a prime and let~$H$ be a 
tree, target graph in \partlabphcol{}.
	Assume $H$ contains a path $\Path= x_0 \dots x_k$ and let $P_k = z_0 \dots z_k$ be the $k$-path of length $k$.
	For every edge $e \in E$, we take a copy of $P_k$ denoted $P_k^e = z_0^e \dots z_k^e$.
	Then, $J$ is constructed starting with $G$ by adding two vertices $\hat{u}$ and $\hat{v}$ and connecting them to 
every vertex in $\lpart$ and $\rpart$, respectively. 
Subsequently, every edge $e \in E$ is substituted with a the path $P_k^e$.
	Finally, the pinning function of $J$ maps~$\hat{u}$ to~$x_0$ as well as~$\hat{v}$ to~$x_k$.
	See Figure~\ref{fig:reduction_from_wBIS} for an example.
	Formally, we have the following definition.
	
	\begin{figure}[t]
		\centering
		\includegraphics[width=\textwidth]{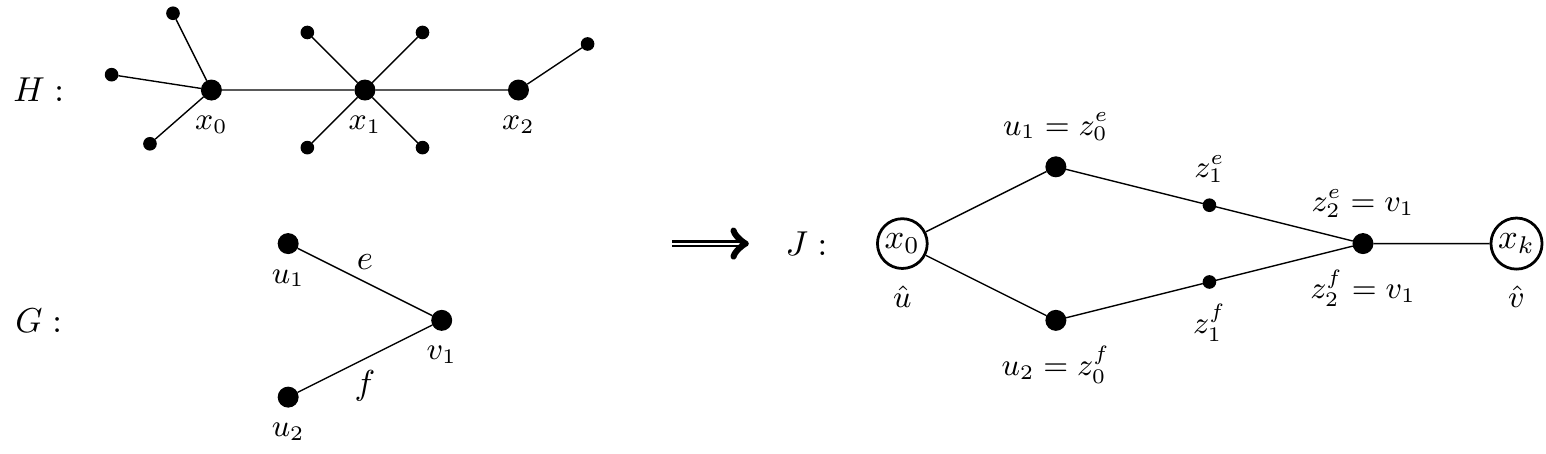}
		\caption{Constructive route for $J$ given $G$ and the $(4,2,5)$-path in $H$ for $p=5$.} \label{fig:reduction_from_wBIS}
	\end{figure}

\begin{definition}\label{def:gadget_wBIS_to_Homs}
	Let $p$ be a prime and $H$ be a graph containing the path $\Path= x_0 \dots x_k$.
	Given the $k$-path $P_k=z_0 \dots z_k$ and a bipartite graph $G=\bipG$. 
Then, $J$ is the partially labelled graph with vertex set 
\begin{align*}
V(G(J)) = \;&\{\,\hat{u}, \,\hat{v}\,\} \cup \lpart \cup \rpart \cup \{\,z_i^e \mid i \in [k-1], e \in E \,\}
\intertext{
and edge set
}
E(G(J)) = \;&\{\,(\hat{u}, u) \mid u \in \lpart \,\} \cup \{\,(z_j^e, z_{j+1}^{e}) \mid e \in E, j \in [k-2] \,\} \\
&\cup \{\, (u, z_1^e), (z_{k-1}^e, v) \mid e=(u, v) \in E \,\} \cup \{\,(v, \hat{v}) \mid v \in \rpart \,\} .
\end{align*}
	Finally, let $\tau(J)=\{\,\hat{u} \mapsto x_0,\,\hat{v}\mapsto x_k\,\}$ be the partial labelling from $G(J)$ to $H$.
\end{definition}

The following lemma studies the properties of $J$, which will help us establish the reduction.
In order to gain these properties, we do need $Q$ to be an $(a,b,p)$-path.

\begin{lemma}\label{lem:tools_reduction_wBIS_to_Homs}
	Let $p$ be a prime, $G=(V_L,V_R,E)$ a bipartite graph and $H$ be a tree.
	Assume there are $a,b\in\Z_p\setminus\{1\}$ such that $H$ contains an $(a,b,p)$-path 
$\Path[H][a,b,p]= x_0 \dots x_k$. We denote the diminished neighbourhoods of $x_0$ and $x_k$ by $\lNeigh = \Gamma_H(x_0)- x_1$ and $\rNeigh =\Gamma_H(x_k)-x_{k-1}$, respectively.
	Additionally, let $J$ be the partially labelled graph according to 
Definition~\ref{def:gadget_wBIS_to_Homs}. Then, for every homomorphism $\sigma$ from $J$ to $H$ the following hold.
	\begin{enumerate}
		\item Let $u \in \lpart$ and $v \in \rpart$, then $\sigma (u) \in \Gamma_H (x_0)$ and $\sigma (v) \in \Gamma_H (x_k)$, respectively;
		\item Let $\mathfrak{O}_\sigma=\{u\in \lpart \mid \sigma(u)=x_1 \} \cup \{v\in \rpart \mid \sigma(v)=x_{k-1}\}$ 
		and $\IO{I}_\sigma= (\lpart \cup \rpart ) \setminus \IO{O}_\sigma$. Given another homomorphism $\sigma'$ from $J$ to $H$, the relation $\rel{\sigma}{\sigma'}$ if $\IO{I}_\sigma = \IO{I}_{\sigma'}$ is an equivalence relation with equivalence class denoted $\IOclass{\cdot}$;
		\item Let $\sigma_1 , \dots , \sigma_\mu$ be representatives from each $\rel{}{}$\!\!~-equivalence class.
		Then, the set $\calI (G)$ of independent sets of $G$ is exactly the set $\{\, \IO{I}_{\sigma_i} \mid i \in [\mu] \,\}$.
		\item For the diminished neighbourhoods holds
		 $|\IOclass{\sigma}| \equiv |\lNeigh|^{|\IO{I}_\sigma \cap \lpart|} |\rNeigh|^{|\IO{I}_\sigma \cap \rpart|}\pmod p$.
	\end{enumerate}
\end{lemma}
\begin{proof}
	We will prove each statement in order.
	\begin{enumerate}[wide, labelwidth=!, labelindent=0pt, parsep=0pt]
		\item 
		We observe that $\tau(J)(\hat{u})=x_0$ and $\hat{u}$ is adjacent to every vertex in $\lpart$. Therefore, $\sigma$ has to map each vertex $u \in \lpart$ to a vertex in the neighbourhood of $x_0$. The analogue argument shows the second result regarding $\hat{v}$ and the neighbourhood of $x_k$.
		\item 
		The statement follows from the observation that each class $\IOclass{\sigma}$ is uniquely determined by the set $\IO{I}_\sigma$.
		\item 
		We commence the proof with establishing that mapping $\sigma$ to $\IO{I}_\sigma$ defines a surjection from $\Homs{J}{H}$ to $\calI(G)$.
		Then we obtain a bijection from $\{\, \IOclass{\sigma_i} \mid i \in [\mu] \,\}$ to $\calI(G)$, as with $\rel{}{}$ we identify exactly the $\sigma$ and $\sigma'$, for which $\IO{I}_\sigma=\IO{I}_{\sigma'}$.
		
		We first argue, that $\IO{I}_\sigma$ is an independent set in $G$ for every $\sigma \in \Homs{J}{H}$.
		Assume towards a contradiction, that there exists $\sigma \in \Homs{J}{H}$ and a pair of vertices $u,v \in \IO{I}_{\sigma}$ with $(u,v) \in E$.
		Without loss of generality let $u \in \lpart$ and $v \in \rpart$.
		Due to property~$\mathit{1}$ and $u,v \in \IO{I}_{\sigma}$ we obtain that $\sigma(u) \in \lNeigh$ and $\sigma(v) \in \rNeigh$.
		Additionally, $H$ is a tree and the path $\sigma (u) x_0 \dots x_k \sigma(v)$ is the unique path connecting $\sigma(u)$ and $\sigma(v)$.
		Therefore, $\sigma(u)$ and $\sigma(v)$ have distance $k+2$ in $H$.
		However, by the construction of $J$ we have $d_{G(J)}(u,v) = k$, which contradicts the existence of $\IOclass{\sigma}$, due to Corollary~\ref{cor:distance}.	
			
		Regarding surjectivity, let $I \in \calI(G)$.
		We are going to define a mapping $\sigma_I$ that is a homomorphism from $J$ to $H$ with $\IO{I}_{\sigma_I}=I$.
		To do so, let $x_{-1} \in W_L$ and $x_{k+1} \in W_R$.
		This is possible, as $\Path[H]$ is a $(a,b,p)$-path and thus we have $W_L \neq \varnothing$ and $W_R \neq \varnothing$.
		Now, let $\sigma_I$ be defined as follows: 
		Because of the pinning we have to map $\hat u$ to $x_0$ and $\hat v$ to $x_k$.
		Further, if $u \in V_L$ is in $I$, then for every $e \in E$ starting with $u$, we map the path $z^e_0 \ldots z^e_k$ to the path $x_{-1}\ldots x_{k-1}$ in $H$.
		On the other hand, for every $v \in V_R \cap I$ and every edge $e=(u,v)$ in $G$, we map the path $z^e_0 \ldots z^e_k$ to $x_{1}\ldots x_{k+1}$.
		If $e \in E$ is such that neither $u$ nor $v$ are in $I$, then we map $z^e_0 \ldots z^e_k$ to $x_{1}\ldots x_{k-1}x_{k}x_{k-1}$.
		By the construction of $J$ it is easy to see, that $\sigma_I \in \Homs{J}{H}$ and $\IO{I}_{\sigma_I}=I$.
	
%
%
		
		\item 
		Let $\tau: J \to H$ be a homomorphism in $\IOclass{\sigma}$.
		We commence with proving that, for every edge $e\in E$,
		\begin{equation}\label{eq:number_homs_paths}
			|\,\Homs{(P_k^e,z^e_0,z^e_k)}{(H,\tau(z^e_0),\tau(z^e_k))}|\equiv 1\pmod p. 
		\end{equation}
		Let $r=|\,\Homs{(P_k^e,z^e_0,z^e_k)}{(H,\tau(z^e_0),\tau(z^e_k))}|$.
		Due to Lemma~\ref{lem:walks}, $r$ is equal to the number of $k$-walks in $H$ from $\tau(z^e_0)$ to $\tau(z^e_k)$.
		First, we observe that the assumption of $\tau(z^e_0)\in W_L$ and $\tau(z^e_k)\in W_R$ yields a contradiction as argued in the proof of property~$\mathit{3}$.
		Subsequently, we assume that $\tau(z^e_0)=x_1$ and $\tau(z^e_k)=x\in W_R$.
		Since $H$ is a tree, $x_1\dots x_k x$ is the only $k$-walk in $H$ between $x_1$ and $x$.
		Now, Lemma~\ref{lem:walks} yields $r=1$.
		Similarly, the assumption of $\tau(z^e_0)=x\in W_L$ and $\tau(z^e_k)=x_{k-1}$ yields $r=1$. 
		
		Finally, we assume $\tau(z^e_0)=x_1$ and $\tau(z^e_k)=x_{k-1}$ and consider the number of $k$-walks in $H$ between $x_1$ and $x_{k-1}$ denoted $W(x_1, x_{k-1}, k)$.
		We recall that $r = W(x_1, x_{k-1}, k)$.
		We denote by $\Path'=x_1 \dots x_{k-1}$ the subpath of $\Path[H][a,b,p]$ connecting $x_1$ and $x_{k-1}$.
		We  derive $\dist_H(x_1,x_{k-1})= k-2$, because $H$ is a tree and $Q'$ is the unique path in $H$ between $x_1$ and $x_{k-1}$.
		Furthermore, every $k$-walk in $H$ between $x_1$ and $x_{k-1}$ can be constructed from $\Path'$ by adding a walk of size $2$ to any vertex $x_i$ in $\Path'$.
		Therefore, every vertex $x_i$ yields one $k$-walk for every vertex in its neighbourhood.
		Since $H$ contains no cycles we only double-counted the walks entirely contained in $\Path'$.
		That is, for every vertex $x_i$ with $2 \leq i \leq k-1$ in $Q'$ the walk revisiting $x_{i-1}$ after reaching $x_i$.
		Removing every such walk once from the calculation yields
		\[
			W(x_1, x_{k-1}, k) = \sum_{i = 1}^{k-1} \deg_H (x_i) - (k - 2).
		\]
		Since $\Path[H][a,b,p]$ is an $(a,b,p)$-path, we obtain, for all $2 \leq i \leq k-1$, that $\deg_H(x_i)\equiv 1\pmod p$ yielding \eqref{eq:number_homs_paths}. 
		
		In order to show property~$\mathit{4}$, we note that also the set $\IO{O}_\sigma$ uniquely determines $\IOclass{\sigma}$.
		Therefore, for any homomorphism $\tau \in \IOclass{\sigma}$ the labelling of vertices in $\IO{O}_\sigma$ as well as $\hat{u}$ and $\hat{v}$ is fixed. 
		Concerning the vertices in $\IO{I}_\sigma$, due to property~$\mathit{1}$ $\tau$ maps a vertex $u \in \IO{I}_\sigma \cap \lpart$ to any vertex in $\lNeigh$ and a vertex $v \in \IO{I}_\sigma \cap \rpart$ to any vertex in $\rNeigh$.
		Due to Definition~\ref{def:gadget_wBIS_to_Homs} of $G(J)$ every vertex $z_0^e$ and $z_k^e$ is identified with a vertex in $\lpart$ and $\rpart$, respectively.
		Finally, due to \eqref{eq:number_homs_paths} once we have fixed a partial labelling $\tau$ of every vertex $z_0^e$ and $z_k^e$ the number of homomorphisms respecting $\tau$ from any path $P^e$ to $H$ is equivalent to $1$ modulo $p$. 
		This establishes the proof of
		\begin{equation*}
			|\IOclass{\sigma}| \equiv|\lNeigh|^{|\IO{I}_\sigma \cap \lpart|} |\rNeigh|^{|\IO{I}_\sigma \cap \rpart|}\pmod p. \qedhere
		\end{equation*}
%
	\end{enumerate}
\end{proof}

Finally, with the above properties at hand we show that the existence of an $(a,b,p)$-path in $H$ yields hardness for 
\partlabphcol{}.

\begin{lemma}\label{lem:Homs_trees_hardcase}
	Let $p$ be a prime and let $H$ be a graph with no automorphism of order~$p$.
	If there are $a, b \in \Z_p\setminus\{1\}$ such that $H$~has an $(a,b,p)$-path $\Path[H][a,b,p]$ then \phcol{} is $\numpp$-hard under Turing reductions.
\end{lemma}
\begin{proof}
	We will show that \pbislr{a-1}{b-1} reduces to \partlabphcol{} under 
	Turing reductions.
	Since $a,b\not\equiv 1\pmod p$, the lemma then results from the Theorems~\ref{thm:wbis-hardness} and \ref{thm:partlabcol}. 
	Let $\Path[H][a,b,p] = x_0 \dots x_k$ and the bipartite graph $G=(V_L,V_R,E)$ be the input for \pbislr{a-1}{b-1}.
	Additionally, let $J$ be the partially labelled graph constructed according to Definition~\ref{def:gadget_wBIS_to_Homs}.
	We observe that every condition of Lemma~\ref{lem:tools_reduction_wBIS_to_Homs} is satisfied.
	
	Let $\sigma_1 , \dots , \sigma_\mu$ be representatives from each $\rel{}{}$~-equivalence class as given by property~$\mathit{3}$ of Lemma~\ref{lem:tools_reduction_wBIS_to_Homs}.
	We obtain
	\begin{equation*}
		\Homs{J}{H} = \sum_{i=1}^\mu |\IOclass{\sigma_i}|, 
	\end{equation*}
	and by property~$\mathit{4}$, for every $i \in [\mu]$, $|\IOclass{\sigma_i}| \equiv|\lNeigh|^{|\IO{I}_ {\sigma_i}\cap \lpart|} |\rNeigh|^{|\IO{I}_{\sigma_i} \cap \rpart|}\pmod p$.
	Additionally, due to Definition~\ref{def:ab-paths} of an $(a,b,p)$-path $|\lNeigh| \equiv a-1 \pmod p$ and $|\rNeigh| \equiv b-1 \pmod p$.
	We deduce	
	\[
		\Homs{J}{H} \equiv  \sum_{i=1}^\mu (a-1)^{|\IO{I}_{\sigma_i} \cap \lpart|} (b-1)^{|\IO{I}_{\sigma_i} \cap \rpart|}\pmod p .
	\]
	Finally, we recall property~$\mathit{3}$ of Lemma~\ref{lem:tools_reduction_wBIS_to_Homs}, which yields the equality of the set $\{ \IO{I}_{\sigma_i} \mid i \in [\mu] \}$ with the set $\calI_G$ of independent sets of $G$.
	This yields
	\begin{align*}
		\Homs{J}{H} &\equiv \sum_{i=1}^\mu (a-1)^{|\IO{I}_{\sigma_i} \cap \lpart|} (b-1)^{|\IO{I}_{\sigma_i} \cap \rpart|}\pmod p \\
		&= \sum_{I \in \calI(G)}(a-1)^{|I \cap \lpart |}(b-1)^{|I \cap \rpart|}.
	\end{align*}
	The latter is exactly the definition of $\wISet{G}[a-1][b-1]$, which concludes the proof.
	\end{proof}

%% file: main_thm.tex
In this section we gather our results into the following dichotomy theorem.

{\renewcommand{\thetheorem}{\ref{thm:modp-trees}}
\begin{theorem}
Let $p$ be a prime and let $H$ be a graph, such that its order~$p$ reduced form~$H^{*p}$ is a tree. If $H^{*p}$ is a 
star, then \phcol{} is computable in polynomial time; otherwise, \phcol{} is $\numpp$-complete.
\end{theorem}}
\begin{proof}
Let $p$ be a prime and $H$ be a graph, such that its order~$p$ reduced form~$H^{*p}$ is a tree. If $H^{*p}$ is a 
complete bipartite graph, then Corollary~\ref{cor:polyt-graphs} yields that $\kkhcol{p}{H^{*p}}$ can be computed in polynomial time.
We note that in this case, $H^{*p}$ has to be a star.
Otherwise, $H^{*p}$ is not a star and by Lemma~\ref{lem:tree_star_or_hard}, $H^{*p}$ contains an $(a,b,p)$-path. 
Lemma~\ref{lem:Homs_trees_hardcase} shows that $\kkhcol{p}{H^{*p}}$ is $\numpp$-hard. The theorem then follows from 
Theorem~\ref{thm:involution-reduction}.
\end{proof}

To justify our title, we use the following proposition showing that our dichotomy theorem holds for all trees. In \cite[Section~5.3]{FJ13} this was stated as an obvious fact, however for the sake of completeness we provide a formal proof.

\begin{proposition}\label{prop:trees_under_automorphisms}
	Let $H$ be a tree and $\varrho$ an automorphism of $H$. Then the subgraph $H^\varrho$ of $H$ induced by the fixed points of $\varrho$ is also a tree.
\end{proposition}
\begin{proof}
	Let $H$ be a tree and $\varrho$ an automorphism of $H$. $H^\varrho$ is a subgraph of $H$, so it suffices to argue for the connectivity of $H^\varrho$.
	Towards a contradiction we assume that $H^\varrho$ is not connected.
	Thus, there	exist two vertices $u,v \in V(H)$, whose image $\varrho(u)$, $\varrho(v)$ are disconnected in $H^\varrho$.
	However, since $H^\varrho$ only contains the fixed points under $\varrho$ we obtain $\varrho(u) = u$ and $\varrho(v)=v$.
	Therefore, there have to exist adjacent vertices $w$, $z$ on the unique path $P$ in $H$ from $u$ to $v$, for which $\varrho(w)$ is connected to $\rho(u)$ but $\rho(z)$ is not.
	The assumption that $\varrho(z)$ is not connected to $u$ results into $\varrho(z)$ not being connected to $\varrho(w)$. This is a contradiction as $\varrho$ has to preserve edges.
\end{proof}

The claim implies that if $H$ is a tree, then its order~$p$ reduced form~$H^{*p}$ is also a tree. This 
yields the following corollary.

{\renewcommand{\thetheorem}{\ref{cor:modp-trees}}
\begin{corollary}
Let $p$ be a prime and let $H$ be a tree. If the order~$p$ reduced form~$H^{*p}$ of $H$ is a star, then \phcol{} is 
computable in polynomial time; otherwise, \phcol{} is $\numpp$-complete.
\end{corollary}
}

To deal with disconnected graphs, Faben and Jerrum~\cite[Theorem~6.1]{FJ13} show the following theorem.
\begin{theorem}[Faben and Jerrum]\label{thm:Fab_Jer_disconnected}
 Let $H$ be a graph that has no automorphism of order~2. If $H'$ is a connected component of $H$ and \parhcol[H'] is $\parp$-hard, then \parhcol{} is $\parp$-hard.
\end{theorem}

The only part where the value 2 of the modulo is required, is the application of their pinning theorem~\cite[Theorem~4.7]{FJ13}. 
Since we have already shown the more general Theorem~\ref{thm:partlabcol}, we conclude that the theorem holds in the 
following form.

\begin{theorem}\label{thm:disconnected}
 Let $p$ be a prime and let $H$ be a graph that has no automorphism of order~$p$. If $H_1$ is a connected component of 
$H$ and \phcol[H_1] $\numpp$-hard, then \phcol{} is $\numpp$-hard.
\end{theorem}

The latter strengthens Theorem~\ref{thm:modp-trees} to the following version.

\begin{theorem}\label{thm:modp_forests}
 Let $H$ be a graph whose order~$p$ reduced form~$H^{*p}$ is a forest. If every component of $H^{*p}$ is a star, \phcol{} is computable in 
polynomial time, otherwise \phcol{} is $\numpp$-complete.
\end{theorem}

%% file: comp_num.tex
We investigate counting homomorphisms modulo a composite integer $k$ and observe that we may restrict our attention to powers of primes.
With this arises the natural question, whether $\#_{p^r}\prb{HomsTo}H$ being computable in polynomial time is equivalent 
to $\#_{p}\prb{HomsTo}H$ being computable in polynomial time, where $p$ is a prime and $r$ a positive integer.
We answer this question negatively, by presenting a graph $H$ for which $\#_2\prb{HomsTo}H$ is computable in polynomial time, while $\#_4\prb{HomsTo}H$ is $\parp$-hard.
This contrasts results by Guo, Huang, Lu and Xia~\cite{GHLX11} on counting constraint satisfaction problems modulo an 
integer. Guo et al. observed that, for every prime $p$ and integer $r$, $\kcsp[{p^r}]$ is computable in polynomial 
time if an only if \pcsp{} is computable in polynomial time.

\medskip
In order to study the complexity of $\#_k\prb{HomsTo}H$ for composite integers $k$, we will use the Chinese 
remainder theorem. Recall that integers $k_1$ and $k_2$ are said to be \emph{relatively prime}, if their only common 
divisor is 1.

\begin{theorem}[Chinese remainder theorem]\label{thm:chinese-remainder}
Let $\{k_i\}_{i=1}^m$ be a pairwise relatively prime family of positive integers, and let $a_1,\dots,a_m$ be arbitrary integers.
Then there exists a solution $a\in\ints$ to the system of congruences \[a\equiv a_i \pmod {k_i}  \quad\quad\quad(i=1,\dots,m).\]
Moreover, any $a'\in\ints$ is a solution to this system of congruences if and only if $a\equiv a'\pmod k$, where $k=\prod_{i=1}^m k_i$.
\end{theorem}

For a proof, see, e.g., \cite[Theorem~17, Chapter~7]{AlgebraBook}.

\begin{lemma}\label{lem:compositeintegers}
Let $k\in\Z_{>0}$ be an integer and $\prod_{i=1}^m k_i$ with $k_i=p_i^{r_i}$ its prime factorisation with primes $p_1,\ldots,p_m$ and positive integers $r_1,\ldots,r_m \in \Z_{>0}$.
If \khcol{} can be solved in polynomial time, then for each $i\in[m]$, $\#_{k_i}\prb{HomsTo}H$ can also be solved in polynomial time.
\end{lemma}
\begin{proof}
Since $k_i$ is a factor of $k$ we take the solution of \khcol{} modulo~$k_i$ and obtain a solution for $\#_{k_i}\prb{HomsTo}H$.
From the Chinese remainder theorem, (Theorem~\ref{thm:chinese-remainder}) the converse is also true:
if for each $i\in[m]$ we can solve $\#_{k_i}\prb{HomsTo}H$ in polynomial time, then we can also solve \khcol{} in polynomial time.
\end{proof}

With this lemma in mind, the subsequent question is, whether $\#_{p^r}\prb{HomsTo}H$ is computable in polynomial time if and only if $\#_{p}\prb{HomsTo}H$ is computable in polynomial time.
Clearly, the first argument in the proof of Lemma~\ref{lem:compositeintegers} shows that if \khcol{} is computable 
in polynomial time then so is \phcol{}, as we can apply the modulo~$p$ operation to a solution of an instance of 
\khcol{}. We will show, by counterexample, that the reverse implication does not hold. Namely we show that for the 
$4$-path $P_4$, $\#_2\prb{HomsTo}P_4$ is computable in polynomial time, while $\#_4\prb{HomsTo}P_4$ is 
$\parp$-hard.
 
\begin{lemma}
Let $P_4$ denote the path $w_1w_2w_3w_4$.
Then $\#_2\prb{HomsTo}P_4$ is computable in polynomial time.
\end{lemma}
\begin{proof}
The function $\rho=\{\,w_1\mapsto w_4, w_4\mapsto w_1, w_2\mapsto w_3, w_3\mapsto w_2\,\}$ is an automorphism of order~2 for $P_4$ without fixed points, so $P_4^{*2}$ is the empty graph. Trivially, for any non-empty input graph~$G$, $\parhcol[P_4^{*2}]$ is always zero.
Thus, $\#_2\prb{HomsTo}P_4$ is computable in polynomial time by Corollary~\ref{cor:polyt-graphs}.
\end{proof}

In the hardness proof of $\#_4\prb{HomsTo}{P_4}$, we will use the following problem as an intermediate stop in our 
chain of reductions.

\probpar{\conbis{k}.}
{Positive integer $k$.}
{Connected graph $G$.}
{$|\calI(G)|\pmod k$.}

Recall Theorem~\ref{thm:bis} showing that \kbis{k} is \nkp{}-complete for all integers~$k$. The next lemma shows that 
\conbis{k} is also hard for all positive integers.

\begin{lemma}
For all integers~$k$, $\conbis{k}$ is $\nkp$-complete.
\end{lemma}
\begin{proof}
We give a Turing reduction from \kbis{k}, then the lemma follows from Theorem~\ref{thm:bis}.
Let $G$ be a bipartite graph, input for $\kbis{k}$.
Assume, without loss of generality, that in the bipartition~$V_L, V_R$ of $V(G)$, all the isolated vertices of $G$ are 
contained in $V_L$. We construct an instance $G'$ for \conbis{k} by adding an extra vertex $v_0$ to a copy of $G$ and 
connecting $v_0$ with all the vertices in $V_L$. 
That is, $V(G')=V(G)\cup\{v_0\}$ and $E(G')=E\cup\{\,(v,v_0), (v_0,v)\mid v\in V_L\,\}$.

We claim that $|\calI(G)|+2^{|V_2|}=|\calI(G')|$.
Let $\calI_1(G')=\{I\in\calI(G')\mid v_0\in I\}$ and let $\calI_2(G')=\{I\in\calI(G')\mid v_0\notin I\}$.
$\calI_1(G')$ and $\calI_2(G')$ partition $\calI(G')$.
For every $I\in\calI_1(G')$, it must be the case that $I\cap V_1=\varnothing$, as every vertex in $V_1$ is adjacent to $v_0$ in $G'$. 
Any subset of $V_2$ can be an independent set in $\calI_1(G')$, hence $|\calI_1(G)|=2^{|V_2|}$.
To conclude the proof of the claim, we will show that $|\calI_2(G')|=|\calI(G)|$.
Since $v_0$ is not in any independent set in $|\calI_2(G')|$, every independent set of $G$ is an independent set in $\calI_2(G')$ and vice versa.
The lemma follows.
\end{proof}

We can now show our claimed hardness result.

\begin{proposition} \label{prop:primepowers}
Let $P_4$ be the path $w_1w_2w_3w_4$. Then $\#_4\prb{HomsTo}P_4$ is $\parp$-hard.
\end{proposition}
\begin{proof}
We are going to show that $\conbis{2}$ reduces to $\#_4\prb{HomsTo}P_4$.
Let $G=(V_L,V_R,E)$ be a non-empty instance of $\conbis{2}$ and let $\calI(G)$ be the set of independent sets of $G$.
We proceed by showing $2|\calI(G)|=|\,\Homs{G}{P_4}|$.

Let $I\in\calI(G)$.
We define $\sigma_I:V(G)\rightarrow V(H)$ to be the following mapping
\[
\sigma_I(v)=
\begin{cases}
 w_1,& \textrm{if }v\in V_L\cap I\\
 w_2,& \textrm{if }v\in V_R\setminus I\\
 w_3,& \textrm{if }v\in V_L\setminus I\\
 w_4,& \textrm{if }v\in V_R\cap I.
\end{cases}
\]

To observe that $\sigma_I$ is a homomorphism, let $(v_1,v_2)\in E$.
We show $(\sigma_I(v_1),\sigma_I(v_2))\in E(P_4)$.
Without loss of generality, assume $v_1$ is in $V_L$, then $\sigma_I(v_1)\in\{w_1,w_3\}$ and since $v_2\in V_R$, we have $\sigma_I(v_2)\in\{w_2,w_4\}$ by the definition of $\sigma_I$.
Assume towards a contradiction $\sigma_I(v_1)=w_1$ and $\sigma_I(v_2)=w_4$.
For the latter to hold, $v_1$ and $v_2$ must both lie in $I$, which is not possible since $I$ is an independent set and $(v_1,v_2)\in E$.
With this we obtain $(\sigma_I(v_1),\sigma_I(v_2))\in E(P_4)$, and therefore $\sigma_I\in\Homs{G}{P_4}$.

Let $\rho=\{w_1\mapsto w_4, w_4\mapsto w_1, w_2\mapsto w_3, w_3 \mapsto w_2\}$ be the automorphism of order~2 of $P_4$.
Clearly, $\rho\circ\sigma_I$ is a homomorphism different from $\sigma_I$, as they differ on all $v \in V(G)$.
Thus, every $I$ yields the two homomorphisms $\sigma_I, \rho\circ\sigma_I \in \Homs{G}{P_4}$.

Let $I, I' \in \calI(G)$, $I \neq I'$, be different independent sets in $G$.
Without loss of generality there is $v\in I\setminus I'$.
For this $v$ all four values $\sigma_I(v)$, $(\rho\circ\sigma_I)(v)$, $\sigma_{I'}(v)$ and $(\rho\circ\sigma_{I'})(v)$ are different, thus $\sigma_I$, $\rho\circ\sigma_I$, $\sigma_{I'}$ and $\rho\circ\sigma_{I'}$ are four different elements of $\Homs{G}{P_4}$.

It remains to argue that for every $\sigma\in\Homs{G}{P_4}$ there is some $I\in\calI(G)$, such that $\sigma=\sigma_I$ or $\sigma=\rho\circ\sigma_I$.
To this end, let $\sigma\in\Homs{G}{P_4}$.
We argue that
$$I_\sigma=\{\,v\in V(G)\mid \sigma(v)\in\{w_1,w_4\}\,\}$$
is an independent set of $G$.
Let $v_1,v_2\in I_\sigma$.
The definition of $I_\sigma$ yields $(\sigma(v_1),\sigma(v_2))\notin E$.
As $\sigma$ is a homomorphism, there can be no edge $(v_1,v_2)\in I_\sigma$, so $I_\sigma$ is an independent set of $G$.
We conclude the proof by showing that $\sigma=\sigma_{I_\sigma}$ or $\sigma=\rho\circ\sigma_{I_\sigma}$.
For, let $v\in V_L\cap I_\sigma$.
If $\sigma(v)=w_1$, then $\sigma=\sigma_{I_\sigma}$, as $G$ is connected.
On the other hand $\sigma(v)=w_4$ implies $\sigma=\rho\circ\sigma_{I_\sigma}$ and the proposition follows.
\end{proof}